\documentclass[journal,twoside,web]{ieeecolor}
\usepackage{generic}
\usepackage{amssymb}
\usepackage{amsmath}
\usepackage{color}
\newtheorem{definition}{Definition}
\newtheorem{assumption}{Assumption}
\newtheorem{theorem}{Theorem}
\newtheorem{remark}{Remark}
\newtheorem{lemma}{Lemma}

\usepackage{cite}
\usepackage{algorithm}
\usepackage{algorithmicx}
\usepackage[noend]{algpseudocode}
\usepackage{bm}
\usepackage{graphicx}
\usepackage{epstopdf}
\usepackage{footmisc}
\usepackage{geometry}
\makeatletter
\newenvironment{breakablealgorithm}
  {
   \begin{center}
     \refstepcounter{algorithm}
     \hrule height.8pt depth0pt \kern2pt
     \renewcommand{\caption}[2][\relax]{
       {\raggedright\textbf{\ALG@name~\thealgorithm} ##2\par}%
       \ifx\relax##1\relax 
         \addcontentsline{loa}{algorithm}{\protect\numberline{\thealgorithm}##2}%
       \else 
         \addcontentsline{loa}{algorithm}{\protect\numberline{\thealgorithm}##1}%
       \fi
       \kern2pt\hrule\kern2pt
     }
  }{
     \kern2pt\hrule\relax
   \end{center}
  }
\makeatother

\title{\LARGE \bf
A Differential Private Method for Distributed Optimization in Directed Networks via State Decomposition 
}
\author{Xiaomeng Chen$^{1}$, Lingying Huang$^{1}$, \thanks{$^{1}$X. Chen, L. Huang and L. Shi are with the Department of Electronic and Computer Engineering, Hong Kong University of Science and Technology, Clear Water Bay, Kowloon, Hong Kong ( email:xchendu@connect.ust.hk, lhuangaq@connect.ust.hk, eesling@ust.hk).} Lidong He$^{2}$, \thanks{$^{2}$L. He is with the school of Automation at Nanjing University of
   Science and Technology,  P. R. China, 210094 (email: lidonghe@njust.edu.cn).} Subhrakanti Dey$^{3}$,\thanks{$^{3}$S. Dey is with the Hamilton Institute, NUI Maynooth, Maynooth, Co. Kildare, Ireland, and also affiliated with Uppsala University. (email: subhrakanti.dey@angstrom.uu.se).} and Ling Shi$^{1}$
}

\begin{document}
\maketitle
\thispagestyle{empty}
\pagestyle{empty}

\begin{abstract}
In this paper, we study the problem of consensus-based distributed optimization,  where a network of agents, abstracted as a directed graph, aims to minimize the sum of all agents' cost functions collaboratively.  In existing distributed optimization approaches (Push-Pull/AB) for directed graphs, all agents exchange their states with neighbors to achieve the optimal solution with a constant stepsize, which may lead to the disclosure of sensitive and private information. For privacy preservation,
 we propose a novel state-decomposition based gradient tracking approach (SD-Push-Pull) for distributed optimzation over directed networks that preserves differential privacy, which is a strong notion that protects agents' privacy against an adversary with arbitrary auxiliary information. The main idea of the proposed approach is to decompose the gradient state of  each agent into two sub-states.  Only one substate is exchanged by the agent with its neighbours over time, and the other one is not shared. That is to say, only one substate is visible to an adversary, protecting the sensitive information from being leaked. It is proved that under certain decomposition principles, a bound for the sub-optimality of the proposed algorithm can be derived and the differential privacy is achieved simultaneously. Moreover, the trade-off between differential privacy and the optimization accuracy is also characterized. Finally, a numerical simulation is provided to illustrate the effectiveness of the proposed approach. 
\end{abstract}
\section{INTRODUCTION}
 
 With the rapid development in networking technologies, distributed optimization over multi-agent networks has been a heated research topic during the last decade, where agents aim to collaboratively minimize the sum of local functions possessed by each agent through  local communication. Compared with centralized ones, distributed algorithms allow more flexibility and scalability  due to its capability of breaking large-scale problems into  sequences of smaller ones. In view of this,  distributed algorithms are  inherently robust to  environment uncertainties and communication failures and are widely adopted in   power grids~\cite{braun2016distributed}, sensor networks \cite{dougherty2016extremum} and vehicular networks \cite{mohebifard2018distributed}.
 
  The most commonly used algorithms for distributed optimization is the Decentralized Gradient Descent (DGD), requiring diminishing step-sizes to ensure optimality \cite{nedic2009distributed}. To overcome this challenge, Xu et al.\cite{xu2015augmented} replaced the local gradient with an estimated global gradient based on the dynamic average consensus \cite{kia2019tutorial} and then proposed a gradient tracking method for distributed optimzation problem. Recently, Pu et al. \cite{pu2020push} and Xin and Khan \cite{xi2018linear}  devised a modified gradient-tracking algorithm called Push-Pull/AB algorithm for consensus-based distributed optimization, which can be applied to a general directed graph including undirected graph as a special case.  
   \\ \indent The above conventional distributed algorithms require each agent to exchange their state information with the neighbouring agent, which is not desirable if the participating agents have sensitive and private information, as the transmitted information is 
	at risk of being intercepted by adversaries.  By hacking into communication links, an adversary may have access to all conveyed  messages, and potentially obtain the private information of each agent by adopting an attack algorithm. The theoretical analysis of privacy disclosure in distributed optimization is presented by Mandal\cite{mandal2016privacy}, where the parameters of cost functions and generation power can be correctly inferred by an eavesdropper in the economic dispatch problem. As the number of privacy leakage events is increasing, there is an urgent need to preserve privacy of each agent in distributed systems.  
	 \\  \indent For the privacy preservation in distributed optimization, there have been several research results. Wang\cite{wang2019privacy} proposed a privacy-preserving average consensus in which the state of an agent is decomposed into two substates.   Zhang et al. \cite{zhang2018enabling} and Lu et al. \cite{lu2018privacy}   
	 combined existing distributed optimization approaches with the partially  homomorphic cryptography. However, these approaches suffer from high computation complexity and communication cost which may be inapplicable for systems with limited resources.  As an appealing alternative, differential privacy has attracted much  attention in light of its rigorous mathematical framework, proven security properties, and easy implementation \cite{nozari2017differentially}. The main idea of differential private approaches is noise perturbation, leading to a tradeoff between privacy and accuracy.
	 Huang et al. \cite{huang2015differentially} devised a differential private distributed optimization algorithm by adding  Laplacian noise on transmitted  message with a decaying stepsize, resulting in a low convergence rate.  A  constant stepsize is achieved by Ding et al. \cite{ding2018consensus,ding2021differentially} where linear convergence is enjoyed by gradient tracking method and differential privacy is achieved by perturbing states.
	 
	  None of the aforementioned approaches, however, is suitable  for directed graphs with weak topological restrictions, which is more practical in real applications.     In practice, the information flows among sensors may not be bidirectional due to the different communication ranges, e.g., the coordinated vechicle control problem \cite{ghabcheloo2005coordinated} and the economic dispatch problem \cite{yang2013consensus}.  To address  privacy leakage  in distributed optimization for agents interacting over an unbalanced graphs, Mao  et al. \cite{mao2020privacy} designed a privacy-preserving algorithm based on the push-gradient method with a decaying stepsize, which is implemented via a case study to the economic dispatch problem. Nevertheless, the algorithm in \cite{mao2020privacy} lacked a formal privacy notion and it cannot achieve differential privacy.\\                                                                               
	\indent All the above motivates us to further develop a differential private distributed optimization algorithm over directed graphs. Inspired by \cite{wang2019privacy},  a novel differential private distributed optimization approach based on state decomposition is proposed for agents communicating over directed networks.  Under the proposed state decomposition mechanism, a Laplacian noise is perturbed on the gradient state and the global gradient is still tracked after state decomposition. 
	The main contributions of this paper are summarized as follows:

\begin{enumerate}
\item We propose a state-decomposition based gradient tracking approach (\textbf{SD-Push-Pull}) for distributed optimziation over unbalanced directed networks, where the gradient state of each agent is decomposed into two substates to maintain the privacy of all agents. Specifically, one sub-state replacing the role of the original state is communicated with neighboring agents while the other substate is not shared. 
 Compared to the privacy-preserving approaches in \cite{huang2015differentially} and \cite{ding2018consensus}, our proposed approach can be applied to more general and practical networks. 

 \item {\color{black}By carefully designing the state decomposition mechanism, we only need to add noise to one substate of directions instead of perturbing both states and directions \cite{ding2021differentially}. Moreover,   our proposed SD-Push-Pull  algorithm does not require the stepsize or noise variance to be diminishing, which ensures a linear convergence 
to a neighborhood of the optimal solution in expectation exponentially fast under a constant stepsize policy (\textbf{Theorem 1}). }

 \item Different from the privacy notion in \cite{wang2019privacy} and \cite{mao2020privacy},  we adopt the definition of differential privacy, which ensures the privacy of agents regardless of any auxiliary information that an adversary may have and enjoys a rigorous formulation. In addition, we prove that the proposed SD-Push-Pull algorithm can achieve $\epsilon$-differential privacy (\textbf{Theorem 2}). 

\end{enumerate}


 \textit{Notations:}  In this paper, $\mathbb{N}$ and $\mathbb{R}$  represent the sets
whose components are natural numbers and real numbers. $\mathbb{R}_{++}$ denotes the set of all positive real numbers. $\mathbb{R}^n$ and $\mathbb{R}^{n\times p}$ represent the set of $n-$dimensional vectors and $n\times p$-dimensional matrices. We let $(\mathbb{R}^{n\times p})^{\mathbb{N}}$ denote the space of vector-valued  sequences in $\mathbb{R}^{n\times p}$. $\mathbf{1}_n\in\mathbb{R}^n$ and $\mathbf{I}_n\in\mathbb{R}^{n\times n}$ represent the vector of ones and the identity matrix, respectively. The spectral radius of matrix $\mathbf{A}$ is denoted by $\rho(\mathbf{A})$. $[\mathbf{x}]_r$ denotes the $r-$th element of the vector $\mathbf{x}$.  For a given constant $\theta>0$, Lap($\theta$) is the Laplace distribution with probability function $f_L(x,\theta) =\frac{1}{2\theta}e^{-\frac{|x|}{\theta}}$. In addition, $\mathbb{E}(x)$ and $P(x)$ denote the expectation and probability distribution of a random variable $x$, respectively. 
 
\section{PRELIMINARIES AND PROBLEM FORMULATION}
\subsection{Network Model}
We consider a group of agents which communicate with each other over a directed graph. The directed graph is denoted as a pair $G\triangleq (\mathcal{N}, \mathcal{E})$, where $\mathcal{N}$ denotes the agents set and $\mathcal{E}\subset \mathcal{N} \times \mathcal{N}$ denotes the edge set, respectively. A communication link from agent $i$ to agent $j$ is denoted by $(j,i)\in \mathcal{ E}$, indicating that agent $i$ can send messages to agent $j$. Given a nonnegative matrix $\mathbf{M}=[m_{ij}]\in\mathbb{R}^{n\times n}$, the directed graph induced by $\mathbf{M}$ is denoted by $\mathcal{G}_{\mathbf{M}}\triangleq (\mathcal{N}, \mathcal{E}_{\mathbf{M}})$, where $\mathcal{N}=\{1,2,\ldots,n\}$ and $(j,i)\in  \mathcal{E}_{\mathbf{M}}$ if and only if $m_{ij}>0$.
The agents who can directly send messages to agent $i$ are represented as in-neighbours of agent $i$ and the set of these agents is denoted as $N_{\mathbf{M},i}^{in}=\{j\in \mathcal{N}\mid (i,j)\in \mathcal{E}_{\mathbf{M}}\}$. Similarly, the agents who can directly receive messages from agent $i$ are represented as out-neighbours of agent $i$ and the set of these agents is denoted as $N_{\mathbf{M},i}^{out}=\{j\in \mathcal{N}\mid (j,i)\in \mathcal{E}_{\mathbf{M}}\}$.

\subsection{Differential Privacy}
Differential privacy serves as a mathematical notion which quantify the degree of the involved individuals' privacy guarantee in a statistical database. We give the following definitions for preliminaries of differential privacy in distributed optimization. 
{\color{black}\begin{definition}\label{df1}
	(Adjacency \cite{ye2021differentially})  Two function sets $\mathcal{S}^{(1)}=\{f_i^{(1)}\}^n_{i=1}$ and  $\mathcal{S}^{(2)}=\{f_i^{(2)}\}^n_{i=1}$, $\mathcal{S}^{(1)}$ and $\mathcal{S}^{(2)}$ are said to be adjacent if there exists some $i_0 \in\{1,2\ldots,n\}$ such that $f_i^{(1)}=f_i^{(2)}, \forall i\neq i_0$ and $f_{i_0}^{(1)}\neq f_{i_0}^{(2)}$.
		
\end{definition}}

Definition \ref{df1} implies that two function sets are adjacent only if one agent changes its objective function. 

\begin{definition}\label{df2}
	(Differential privacy \cite{dwork2006calibrating}) Given $\epsilon>0,$ for any pair of adjacent function sets $\mathcal{S}^{(1)}$ and $\mathcal{S}^{(2)}$ and any observation $\mathcal{O}\subseteq \text{Range}(\mathcal{A})$, a randomized algorithm $\mathcal{A}$ keeps $\epsilon-$differentially private if
	$$P\{\mathcal{A}(\mathcal{S}^{(1)})\in \mathcal{O}\}\leq e^\epsilon P\{\mathcal{A}(\mathcal{S}^{(2)})\in\mathcal{O}\}$$ 
	where $\text{Range}(\mathcal{A})$ denotes the output codomain of $\mathcal{A}$.
	\end{definition}
	
	Definition \ref{df2} illustrates that
	 a random mechanism is differentially private if its outputs are nearly statistically identical over two similar inputs which only differ in one element. Hence, an eavesdropper cannot distinguish between two function sets with high probability  based on the output of the mechanism. 
Here, a smaller $ \epsilon $ represents a higher level of privacy since the eavesdropper has less chance to distinguish sensitive information of each agent from the observations. Nevertheless, a high privacy level will sacrifice the accuracy of the optimization algorithm. Hence, the constant $\epsilon$ determines a tradeoff between the privacy level and the accuracy. 

\subsection{Problem Formulation}
Consider an optimization problem in a multi-agent system of $n$ agents. Each agent has a private cost function $f_i$, which is only known to agent $i$ itself. All the participating agents aim to minimize a global objective function
	\begin{equation}\label{pro1}
	\min\limits_{x\in\mathbb{R}^{p}}\sum\limits_{i=1}^{n}f_{i}(x)
	\end{equation}
	where $x$  is the global decision variable.
	
	To solve Problem \eqref{pro1}, assume each agent $i$ maintains a local copy of $x_i \in \mathbb{R}^{p}$ of the decision variable and an auxiliary variable $y_i \in \mathbb{R}^{p}$ tracking the average gradients. Then we can rewrite Problem \eqref{pro1} into local optimization problem of each agent with an added consensus constraint as follows
	\begin{equation}\label{pro2}
	\begin{split}
	\min_{x_i \in \mathbb{R}^p} \quad & \sum_{i=1}^n f_i(x_i)   \\
	\text{s.t.} \quad & x_i = x_j, \quad \forall i,j, \\
	\end{split}
	\end{equation} 

	where $ x_i $ is the local decision variable of the agent $ i $.
	
	Let
	$$\begin{aligned}
&\mathbf{x}:=[x_1,x_2,\ldots,x_n]^\top
\in\mathbb{R}^{n\times p},\\
&\mathbf{y}:=[y_1,y_2,\ldots,y_n]^\top
\in\mathbb{R}^{n\times p}.
	\end{aligned}
	$$
	
	Denote $F(x)$ as an aggregate objective function of the local variables, i.e.,   $F(x)=\sum_{i=1}^n f_i(x_i)$.
	
	With respect to the objective function in Problem \eqref{pro1}, we assume the following strong convexity and smoothness conditions. 
	\begin{assumption}\label{asp}
Each objective function $f_i$ is $\mu-$strongly convex with $L-$Lipschitz continuous gradients, i.e., for any $\mathbf{x}$, $\mathbf{y}\in\mathbb {R}^p$,
{\color{black}
$$\begin{aligned}
	&\langle\nabla f_i(\mathbf{x})- \nabla f_i(\mathbf{y}), \mathbf{x}-\mathbf{y}\rangle \geq \mu||\mathbf{x}-\mathbf{y}||^2,\\
	&||\nabla f_i(\mathbf{x})- \nabla f_i(\mathbf{y})||\leq L||\mathbf{x}-\mathbf{y}||.\\
\end{aligned}
$$}
\end{assumption}
\vspace*{2mm}

Under Assumption \ref{asp}, Problem \eqref{pro1} has a unique optimal solution $x^{\star}\in\mathbb{R}^{p}$ \cite{pu2018push}.

\section{PRIVATE GRADIENT TRACKING ALGORITHM VIA STATE DECOMPOISTION}
In this section, we propose a state-decomposition based  gradient tracking method (SD-Push-Pull) for distributed optimization  over directed graphs, which is described in Algorithm \ref{alg1}. 
 The main idea  is to let each agent decompose its gradient state $y_{i}$ into two substates $y^\alpha_{i}$ and $y^\beta_{i}$. The substate  $y^\alpha_{i}$ is used in the communication with other agents while $y^\beta_{i}$ is never shared with other agents except for agent $i$ itself, so the substate $y^\beta_{i}$ is imperceptible to the neighbouring agents of agent $i$.

 \begin{breakablealgorithm}
\caption{SD-Push-Pull}
\label{alg1}
\begin{algorithmic}
 \State \textbf {Step 1}.  Initialization:
 \begin{enumerate}
 \item Agent $i\in \mathcal{N}$ chooses in-bound mixing/pulling weights $R_{ij}\geq 0$ for all $j\in N^{in}_{R,i}$, out-bound pushing weights $\tilde C_{li}\geq 0$ for all $l\in N^{out}_{\tilde C,i}$, and the two sub-state weights $\alpha_i, \beta_i  \in (0,1)$.
\item Agent $i\in \mathcal{N}$ picks any {\color{black}$x_{i,0}, y_{i,0}^\alpha,  y_{i,0}^\beta\in \mathbb{R}^p,$  $\theta _i \in \mathbb{R_{++}}$}.
 \item The step size $\eta>0$ is known to each agent.  

 \end{enumerate} 

  \textbf {Step 2.} At iteration $k=0,1,2,\ldots,{\color{black}K}$\\
   {\color{black}
 \begin{enumerate}

 \item  Agent $i\in \mathcal{N}$ pushes $\tilde C_{li}y_{i,k}^\alpha$ to each $l\in N^{out}_{\tilde C,i}$.
  \item Agent $i\in \mathcal{N}$ draws a random vector $\xi_{i,k}$ consisting of $p$  Laplacian noise independently drawn from Lap($\theta_i$) and updates $y_{i,k+1}^\alpha$, $y_{i,k+1}^\beta$ as follows:
\begin{subequations}\label{y}
 	\begin{align}\label{ya}
 		&y_{i,k+1}^\alpha=\sum_{j\in N^{in}_{\tilde C,i}\cup\{i\}} \tilde C_{ij}y_{j,k}^\alpha+(1-\beta_i) y_{i,k}^\beta+\xi_{i,k}\\\label{yb}
 		&y_{i,k+1}^\beta=\alpha_i y_{i,k}^\alpha+\beta_i y_{i,k}^\beta+\nabla f_i(x_{i,k})
 		 	\end{align}
 	 \end{subequations}

  \item Agent $i\in \mathcal{N}$ pulls  $[x_{j,k}-\eta (y_{j,k+1}^\alpha-y_{j,k}^\alpha)]$ from each $j\in  N^{in}_{R,i}$.	
\item Agent $i\in \mathcal{N}$ updates $x_{i,k+1}$ through
\begin{equation}\label{d1}
	x_{i,k+1}=\sum_{j\in N^{in}_{R,i}\cup\{i\}} R_{ij}[x_{j,k}-\eta (y_{j,k+1}^\alpha-y_{j,k}^\alpha)].
\end{equation}
\end{enumerate}}
\end{algorithmic}
\end{breakablealgorithm}
 
 Denote $$\begin{aligned}
 	&\mathbf{R}:=[R_{ij}], \quad \mathbf{\Lambda_\alpha}:=diag(\alpha_1,\alpha_2,\ldots,\alpha_n), \\
 	&\mathbf{\tilde C}:=[\tilde C_{ij}],\quad \mathbf{\Lambda_\beta}:=diag(\beta_1,\beta_2,\ldots,\beta_n), \\
&{\color{black}\mathbf{\xi}_k:=[\xi_{1,k},\xi_{2,k}\ldots, \xi_{n,k}]^\top},\\
&\mathbf{y}_k:=[y_{1,k}^\alpha,\ldots,y_{n,k}^\alpha(k),y_{1,k}^\beta(k),\ldots,y_{n,k}^\beta(k)]^\top\\
&\nabla F(\mathbf{x}_k)=[\nabla f_1(x_{1,k}),\ldots,\nabla f_n(x_{n,k})]^\top,\\
&{\color{black}\nabla {\tilde F}(\mathbf{x}_{k})=[\mathbf{\xi}_k^\top, \nabla F(\mathbf{x}_k)^\top]^\top,}\\
 	&\mathbf{ C}:=\begin{bmatrix}
 		\mathbf{\tilde C}&\bf{I}_n-\mathbf{\Lambda_\beta}\\
 		\mathbf{\Lambda_\alpha}&\mathbf{\Lambda_\beta}\\
 		 	\end{bmatrix},\quad \mathbf{T}:=\begin{bmatrix}
 		\mathbf{I}_n & \mathbf{0}_n
 	\end{bmatrix}.\\
 \end{aligned}
 $$
 
Algorithm \ref{alg1} can be rewritten
 in a matrix form as follows:
 {\color{black}\begin{subequations}\label{eqalg1}
\begin{align}
 	&\mathbf{y}_{k+1}=\mathbf{C}\mathbf{y}_k+\nabla {\tilde F}(\mathbf{x}_{k}),\\
 	&\mathbf{x}_{k+1}=\mathbf{R}[\mathbf{x}_k-\eta \mathbf{T} (\mathbf{y}_{k+1}-\mathbf{y}_k)],
\end{align}	
where $\mathbf{x}_0$ and $\mathbf{y}_0$ are arbitrary.
 \end{subequations}}
 {\color{black}
 \begin{remark}
In SD-Push-Pull, we design a state decomposition mechanism for the direction state (shown in \eqref{y}), where  sensitive information $\nabla f_i(x_{i,k})$ is contained in one substate of direction, $y_{i,k+1}^\beta$. Since the substate $y_{i,k+1}^\beta$ is not shared in  communication link, the private information    $\nabla f_i(x_{i,k})$ of agent $i$ is protected from being leaked. Although $y_{i,k}^\beta$ will be shared through $y_{i,k+1}^\alpha$,  the  noise $\xi_{i,k}$ added in the updated of  $y_{i,k}^\alpha$ (shown in \eqref{ya}) helps to avoid privacy breaches of $f_i$. The above discussion intuitively illustrates  how state decomposition mechanism achieves differential privacy of each agent. Rigorous theoretical  analysis will be provided in Section \ref{privacy}. 

 \end{remark}
}

\begin{assumption}\label{asp1}
	The matrix $\mathbf{R}\in\mathbb{R}^{n\times n}$ is a nonnegative row-stochastic matrix and   $\mathbf{C}\in\mathbb{R}^{2n\times 2n}$ is a nonnegative column-stochastic matrix, i.e., $\mathbf{R1}_n=\mathbf{1}_n$ and $\mathbf{1}_{2n}^\top \mathbf{C}=\mathbf{1}_{2n}^\top$. Moreover, the diagonal entries of $\mathbf{R}$ and  $\mathbf{C}$ are positive, i.e., $R_{ii}>0, \tilde C_{ii}>0, \forall i\in \mathcal{N}$.
\end{assumption}

Assumption \ref{asp1} can be satisfied by properly designing the weights in $\mathbf{R}$ and $\mathbf{C}$ by each agent locally. For instance, each agent may choose $R_{ij}=\frac{1}{| N^{in}_{R,i}|+c_R}$ for some constant $c_R>0$ for all $j\in N^{in}_{R,i}$ and let $R_{ii}=1-\sum_{j\in N^{in}_{R,i}}R_{ij}$. {\color{black}Similarly, agent $i$ may choose $\alpha_i=\zeta$ and  $\tilde C_{li}=\frac{1-\zeta}{| N^{out}_{C,i}|+c_C}$ for some constant $0<\zeta<1, c_C>0$ for all $l\in N^{out}_{C,i}$, and let $\tilde C_{ii}=1-\zeta-\sum_{l\in N^{out}_{C,i}} \tilde C_{li}$. }Such a choice of weights renders  $\mathbf{R}$ row-stochastic and $\mathbf{C}$ 
 column-stochastic, thus satisfying Assumption \ref{asp1}. 
 \begin{assumption}\label{asp2}
	The graphs $\mathcal{G}_\mathbf{R}$ and $\mathcal{G}_\mathbf{C^\top}$ induced by matrices $\mathbf{R}$ and $\mathbf{C}$ contain at least one spanning tree. In addition, there exists at least one agent that is a root of spanning trees for both  $\mathcal{G}_\mathbf{R}$ and $\mathcal{G}_\mathbf{C^\top}$.
\end{assumption}

{\color{black} Assumption \ref{asp2} is weaker than  assumptions in most previous works (e.g.,\cite{xi2018linear}, \cite{nedic2017achieving},  \cite{xin2018linear}), where graphs $\mathcal{G}_\mathbf{R}$ and $\mathcal{G}_\mathbf{C^\top}$ are assumed to be strongly connected. The relaxed assumption about graph topology enables us to design graphs $\mathcal{G}_\mathbf{R}$ and $\mathcal{G}_\mathbf{C^\top}$ more flexibly. Similar assumption are adopted in \cite{pu2018push}, \cite{pu2020robust}.}

\begin{lemma}[\cite{horn2012matrix}]\label{lem1}
	 Under Assumption \ref{asp1}-\ref{asp2}, the matrix $\mathbf{R}$ has a unique nonnegative left eigenvector $u^\top$ (w.r.t. eigenvalue) with $u^\top \mathbf{1}_n=n,$ and  matrix $\mathbf{C}$ has a unique nonnegative right eigenvector $v$ (w.r.t. eigenvalue) with $ \mathbf{1}_{2n}^\top v=n$.
\end{lemma}
 
 {\color{black}To make connections with the traditional push-pull algorithm, we let $ \mathbf{\tilde y}_k:= \mathbf{ y}_{k+1}-\mathbf{ y}_{k}$. Then one can  have
 \begin{subequations}\label{eqalg}
\begin{align}\label{xstate}
 	&\mathbf{x}_{k+1}=\mathbf{R}(\mathbf{x}_k-\eta \mathbf{T} \mathbf{\tilde y}_k),\\
 	&\mathbf{\tilde y}_{k+1}=\mathbf{C}\mathbf{\tilde y}_k+\nabla {\tilde F}(\mathbf{x}_{k+1})-\nabla {\tilde F}(\mathbf{x}_{k}).
\end{align}	\vspace*{-2mm}
 \end{subequations}
 \vspace*{-2mm}
 
 }
Since $C$ is column stochastic, \vspace*{-2mm}
$$\mathbf{1}_{2n}^\top \mathbf{\tilde y}_{k+1}=\mathbf{1}_{2n}^\top \mathbf{\tilde y}_{k}+\mathbf{1}_{2n}^\top \nabla {\tilde F}(\mathbf{x}_{k+1})-\mathbf{1}_{2n}^\top \nabla {\tilde F}(\mathbf{x}_{k}). $$ 

{\color{black}From  $\mathbf{1}_{2n}^\top\mathbf{\tilde y}_0= \mathbf{1}_{2n}^\top\mathbf{ y}_{1}-\mathbf{1}_{2n}^\top\mathbf{ y}_{0}=\mathbf{1}_{2n}^\top\nabla \tilde F(\mathbf{x}_{0})$, we have by induction that \vspace*{-2mm}
\begin{equation}\label{sumy}
	\frac{1}{n}\mathbf{1}_{2n}^\top \mathbf{\tilde y}_{k}={\color{black}\frac{1}{n}\mathbf{1}_{2n}^\top} \nabla \tilde F(\mathbf{x}_{k})=\frac{1}{n}\mathbf{1}_{n}^\top \nabla  F(\mathbf{x}_{k})+\frac{1}{n}\mathbf{1}_{n}^\top \mathbf{\xi}_k.
\end{equation}
}
\vspace*{-2mm}

 \section{Convergence Analysis}
 In this section, we analyze the convergence performance of the proposed private push-pull algorithm. For the sake of analysis, we define the following variables:
$$\bar x_k:=\frac{1}{n}u^\top \mathbf{x}_k,\qquad \bar y_k=\frac{1}{n}\mathbf{1}_{2n}^\top \mathbf{\tilde y}_k.$$

The main idea of our strategy is to bound $\mathbb{E}[||\bar x_{k+1}-x^{\star\top}||]_2, \mathbb{E}[||\mathbf{x}_{k+1}-\mathbf{1}_n\bar x_{k+1}||_R], \mathbb{E}[||\mathbf{\tilde y}_{k+1}-v\bar y_{k+1}]||_C]$ on the basis of the linear combinations of their previous values, where $||\cdot||_R$ and $||\cdot||_C$ are specific norms to be defined later. By establishing a linear system of inequalities, we can derive the convergence result. 
\begin{definition}
	Given an arbitrary vector norm $||\cdot||$, for any $\mathbf{x}\in\mathbb{R}^{n\times p}$, we {\color{black}define a matrix norm}
	$$||\mathbf{x}||:=\Big|\Big|\Big[||\mathbf{x}^{(1)}||,||\mathbf{x}^{(2)}||,\ldots, ||\mathbf{x}^{(p)}||\Big]\Big|\Big|_2
	$$
	where $\mathbf{x}^{(1)},\mathbf{x}^{(2)},\ldots,\mathbf{x}^{(p)}\in\mathbb{R}^n$ are columns of $\mathbf{x}$.
\end{definition}
\subsection{Preliminary Analysis} 

From  Eqs. \eqref{xstate} and Lemma \ref{lem1}, we can obtain
\begin{equation}\label{p1}
	\bar x_{k+1}=\frac{1}{n}u^\top \mathbf{R}(\mathbf{x}_k-\eta \mathbf{T} \mathbf{\tilde y}_k)=\bar x_{k}-\frac{\eta}{n}u^\top\mathbf{T} \mathbf{\tilde y }_k.
\end{equation}

Furthermore, let us define {\color{black}$$g_k:=\frac{1}{n}\sum_{i=1}^n\nabla f_i(\bar x_{k}),\quad h_k:=\frac{1}{n}\sum_{i=1}^n \nabla f_i( x_{i,k}),$$}

which leads to
$
	\bar y_k=h_k+\frac{1}{n}\mathbf{1}_{n}^\top \mathbf{\xi}_k.
$
Then, from equation \eqref{p1} 
\begin{equation}\label{m1}
\begin{aligned}
	\bar x_{k+1}=&\bar x_{k}-\frac{\eta}{n}u^\top\mathbf{T} (\mathbf{\tilde y}_k-v\bar y_k+v\bar y_k)\\
	=&\bar x_{k}-\frac{\eta}{n}u^\top\mathbf{T}v\bar y_k-\frac{\eta}{n}u^\top\mathbf{T}(\mathbf{\tilde y}_k- v\bar y_k)\\
	=&{\color{black}\bar x_{k}-\eta'g_k-\eta'(h_k-g_k)-\frac{\eta}{n}u^\top\mathbf{T}(\mathbf{\tilde y}_k- v\bar y_k)}\\
	&{\color{black}-\frac{\eta'}{n}\mathbf{1}_{n}^\top \mathbf{\xi}_k,}\\
\end{aligned}	
\end{equation}
where $\eta':=\frac{\eta}{n}u^\top\mathbf{T}v.$

Based on Lemma \ref{lem1} and equation \eqref{p1}, we obtain
\begin{equation}\label{m2}\vspace*{-1mm}
\begin{aligned}
		&\mathbf{x}_{k+1}-\mathbf{1}_{n}\bar{x}_{k+1}\\
		&=\mathbf{R}(\mathbf{x}_k-\eta \mathbf{T} \mathbf{\tilde y}_k)-\mathbf{1}_{n}\bar x_{k}+\frac{\eta}{n}\mathbf{1}_{n}u^\top\mathbf{T} \mathbf{\tilde y}_k\\
		&=\mathbf{R}(\mathbf{x}_k-\mathbf{1}_{n}\bar x_{k})-(\mathbf{R}-\frac{\mathbf{1}_{n}u^\top}{n})\eta\mathbf{T} \mathbf{\tilde y}_k\\
		&=(\mathbf{R}-\frac{\mathbf{1}_{n}u^\top}{n})(\mathbf{x}_k-\mathbf{1}_{n}\bar x_{k})-(\mathbf{R}-\frac{\mathbf{1}_{n}u^\top}{n})\eta\mathbf{T} \mathbf{\tilde y}_k.
\end{aligned}\end{equation}

Similarly, we have
{\color{black}
\begin{equation}\label{m3}
	\begin{aligned}
		&\mathbf{\tilde y}_{k+1}-v\bar{y}_{k+1}\\
	&=\mathbf{C}\mathbf{\tilde y}_k+\nabla \tilde F(\mathbf{x}_{k+1})-\nabla \tilde F(\mathbf{x}_{k})-v(\bar{y}_{k+1}-\bar{y}_{k})-v\bar{y}_{k+1}\\
		&=(\mathbf{C}-\frac{v\mathbf{1}_{2n}^\top}{n})(\mathbf{\tilde y}_k-v\bar{y}_{k})\\
		&\qquad +(\mathbf{I}_{2n}-\frac{v\mathbf{1}_{2n}^\top}{n})(\nabla \tilde F(\mathbf{x}_{k+1})-\nabla \tilde F(\mathbf{x}_{k})).\\
	\end{aligned}
\end{equation}
}
Denote  $\mathcal{F}_k$ as the $\sigma$-algebra generated by
${\color{black}\{\mathbf{\xi}_0,\ldots, \mathbf{\xi}_{k-1}\}}$, and define $\mathbb{E}[\cdot | \mathcal{F}_k]$ as the conditional expectation given $\mathcal{F}_k$. 
\subsection{Supporting lemmas}
We next prepare a few useful supporting lemmas for further convergence analysis. 
\begin{lemma}\label{lemma2}
	Under Assumption \ref{asp}, there holds
	{\color{black}
	$$\begin{aligned}
		&||h_k-g_k||_2\leq \frac{L}{\sqrt{n}}||\mathbf{x}_k-\mathbf{1}_n \bar x_k||_2,\\
		&||g_k||_2\leq L||\bar x_k-x^{\star\top}||_2.\\
		&\mathbb{E}[||\bar y_k-h_k||_2^2\mid \mathcal{F}_k]\leq \frac{2p\bar \theta^2}{n},\\
	\end{aligned}$$
	where $\bar \theta =\max\limits_i{\theta_i}$.}
\end{lemma}
\vspace*{2mm}
\begin{proof}
	In view of Assumption \ref{asp}, 
	{\color{black}$$\begin{aligned}
			&||h_k-g_k||_2=\frac{1}{n}||\mathbf{1}_n^\top\nabla F(\mathbf{x}_k)-\mathbf{1}_n^\top\nabla F(\mathbf{1}_n{\bar x}_k)||_2\\
			&\leq \frac{L}{n}\sum_{i=1}^n||x_{i,k}-\bar x_k||_2\leq \frac{L}{\sqrt{n}}||\mathbf{x}_k-\mathbf{1}_n{\bar x}_k||_2,
				\end{aligned}$$}
				
				$$\begin{aligned}
					||g_k||_2=&\frac{1}{n}||\mathbf{1}_n^\top\nabla F(\mathbf{1}_n{\bar x}_k)-\mathbf{1}_n^\top\nabla F(\mathbf{1}_n{\bar x}^\star)||	\\
					&\leq \frac{L}{n}\sum_{i=1}^n||\bar x_k-x^{\star\top}||_2=L		||\bar x_k-x^{\star\top}||_2, 	
\end{aligned}$$	
				
{\color{black}$$\mathbb{E}[||\bar y_k-h_k||_2^2\mid \mathcal{F}_k]=\frac{1}{n^2}\mathbb{E}[||\mathbf{1}_n^\top\mathbf{\xi_k}||_2^2\mid \mathcal{F}_k]\leq \frac{2p\bar \theta^2}{n},$$
where $\bar \theta =\max_i{\theta_i}$.}
\vspace*{-0.5mm}
\end{proof}

\begin{lemma}\label{lemma3}
	(Adapted from Lemma 10 in \cite{qu2017harnessing}) Under Assumption \ref{asp}, for any $x\in \mathbb{R}^p$ and $0<\theta<2/\mu,$ we have\vspace*{-1mm}
	$$||x-\theta F(x)-x^{\star\top}||_2\leq 
\tau ||x-x^{\star\top}||,\vspace*{-1mm}$$
where $\tau=\max (|1-\mu \theta|,|1-L \theta|)$.
\end{lemma}
\begin{lemma}\label{lemma4}
	(Adapted from Lemma 3 and Lemma 4 in \cite{pu2020push}) Suppose Assumption \ref{asp1} and \ref{asp2} hold. There exist vector norms $||\cdot||_R$ and $||\cdot||_{C}$, such that $\sigma_R:=||\mathbf{R}-\frac{\mathbf{1}_n u^\top}{n}||_R<1, \sigma_C:=||\mathbf{C}-\frac{v  \mathbf{1}_{2n}^\top }{n}||_C<1, $ and $\tau_R$ and $\tau_C$ are arbitrarily close to the spectral radii $\rho(\mathbf{R}-\mathbf{1}_n u^\top/n)<1$ and $\rho(\mathbf{C}-v  \mathbf{1}_{2n}^\top /n)<1,$. 
\end{lemma}

The following two lemmas are also taken from \cite{pu2018push}. 
\begin{lemma}\label{lemma5}
	Given an arbitrary norm $||\cdot||$, for $\mathbf{W}\in\mathbb{R}^{m\times n}$ and $\mathbf{x}\in\mathbb{R}^{n\times p}$, we have $||\mathbf{Wx}||\leq ||\mathbf{W}||||\mathbf{x}||$. For any $w \in\mathbb{R}^{n\times 1}$ and $x\in\mathbb{R}^{1\times p}$, we have $||wx||=||w||||x||_2$. 
\end{lemma}

\begin{lemma}\label{lemma6}
	There exists constants $\delta_{C,R},\delta_{C,2},\delta_{R,C},\delta_{R,2}>0,$  such that $||\cdot||_C\leq \delta_{C,R}||\cdot||_R$, $||\cdot||_C\leq \delta_{C,2}||\cdot||_2$,$||\cdot||_C\leq \delta_{C,R}||\cdot||_R$, $||\cdot||_R\leq \delta_{R,C}||\cdot||_C$,$||\cdot||_R\leq \delta_{R,2}||\cdot||_2$. Moreover, with a proper rescaling of norms $||\cdot||_R$ and  $||\cdot||_C$, we have $||\cdot||_2\leq ||\cdot||_R$ and $||\cdot||_2\leq ||\cdot||_C$.

\end{lemma}

\begin{lemma}\label{lemma8}
	(Lemma 5 in \cite{pu2020distributed}) Given a nonnegative, irreducible matrix $\mathbf{M}= [m_{ij}] \in \mathbf{R}^{3\times3}$ with its diagonal element $m_{11}, m_{22}, m_{33}<\lambda^\star$  for some $\lambda^\star > 0$. A necessary and sufficient condition for $\rho(\mathbf{M}) < \lambda^\star$ is det$(\lambda^\star\mathbf{I}-\mathbf{M})>0$.
\end{lemma}

	\begin{lemma}\label{lm8}
		 For any $\mathbf{U,V} \in \mathbb{R}^{n\times n}$, the following inequality is satisfied:\vspace*{-2mm}
		\begin{equation}\vspace*{-1mm}
			||\mathbf{U}+\mathbf{V}||^2\leq \tau'||\mathbf{U}||^2+\frac{\tau'}{\tau'-1}||\mathbf{V}||^2,\vspace*{-1mm}
		\end{equation}
		where $\tau'>1$. Moreover, for any $\mathbf{U}_1, \mathbf{U}_2, \mathbf{U}_3 \in \mathbb{R}^{n\times n}$, we have $||\mathbf{U}_1+\mathbf{U}_2+ \mathbf{U}_3||^2\leq \tau_1 ||\mathbf{U}_1||^2+\frac{2\tau_1}{\tau_1-1}(||\mathbf{U}_2||^2+||\mathbf{U}_3||^2)$ and $||\mathbf{U}_1+\mathbf{U}_2+ \mathbf{U}_3||^2\leq 3 ||\mathbf{U}_1||^2+3||\mathbf{U}_2||^2+3||\mathbf{U}_3||^2$. 
	\end{lemma}
\subsection{Main results}
The following critical lemma establishes a linear system of inequalities that  bound $\mathbb{E}[||\bar x_{k+1}-x^{\star\top}||_2^2], \mathbb{E}[||\mathbf{x}_{k+1}-\mathbf{1}_n\bar x_{k+1}||_R^2],$ and $\mathbb{E}[||\mathbf{\tilde y}_{k+1}-v\bar y_{k+1}||_C^2]$.

\begin{lemma}\label{mainl}
	Under Assumptions \ref{asp}-\ref{asp2}, when $\eta'<1/(\mu+L)$, we have the following linear system of inequalities:
	\begin{equation}\label{ine}
	\begin{aligned}
&\begin{bmatrix}
	\mathbb{E}[||\bar x_{k+1}-x^{\star\top}||_2^2|\mathcal{F}_k]\\
	\mathbb{E}[||\mathbf{x}_{k+1}-\mathbf{1}_n\bar x_{k+1}||_R^2|\mathcal{F}_k]\\
	\mathbb{E}[||\mathbf{\tilde y}_{k+1}-v\bar y_{k+1}]||_C^2|\mathcal{F}_k]
\end{bmatrix}\leq \mathbf{A}\begin{bmatrix}
	\mathbb{E}[||\bar x_{k}-x^{\star\top}||_2^2|\mathcal{F}_k]\\
	\mathbb{E}[||\mathbf{x}_{k}-\mathbf{1}_n\bar x_{k}||_R^2|\mathcal{F}_k]\\
	\mathbb{E}[||\mathbf{\tilde y}_{k}-v\bar y_{k}]||_C^2|\mathcal{F}_k]
\end{bmatrix}\\
&\qquad \qquad \qquad \qquad \qquad \qquad \qquad \qquad \qquad \qquad \qquad +\mathbf{b},		
	\end{aligned}
	\end{equation}	where the inequality is taken component-wise, and elements of the transition matrix $\mathbf{A} = [a_{ij}]$ and the vector $\mathbf{b}$ are given by:
\begin{equation}
	\mathbf{A}=\begin{bmatrix}
		1-\eta' \mu & c_1\eta & c_2\eta\\
		c_4\eta^2 & \frac{1+\sigma_R^2}{2}+c_5\eta^2 & c_6\eta^2 \\
		c_8 \eta^2 & c_9 &\frac{1+\sigma_C^2}{2}+c_{10}\eta^2\\
	\end{bmatrix},
\end{equation}
and \begin{equation}
	\mathbf{b}=\begin{bmatrix}
		c_3\eta^2& c_7\eta^2&c_{11}
	\end{bmatrix}^\top \bar\theta^2,
		\end{equation}
		respectively, where constants $c_i$'s and $b_1$ are defined in \eqref{c1}, \eqref{c2} and \eqref{c3}.
		\end{lemma}
\vspace*{2mm}
\begin{proof}
   See Appendix	\ref{ap1}. 
\end{proof}

The following theorem shows the convergence properties for the SD-Push-Pull algorithm in \eqref{eqalg}.
\begin{theorem}\label{th1}
	Suppose Assumption \ref{asp}-\ref{asp2} holds and the stepsize $\eta$ satisfies
	$$\eta\leq\min\big\{\sqrt{\frac{1-\sigma_R^2}{6c_5}}, \sqrt{\frac{1-\sigma_C^2}{6c_{10}}},\sqrt{\frac{2d_3}{d_2+\sqrt{d_2^2+4d_1d_3}}}
\big\},$$
	where $d_1,d_2,d_3$ are defined in \eqref{ds}. Then $\sup_{l\geq k} \mathbb{E}[||\bar x_{l}-x^{\star\top}||_2^2]$ and $\sup_{l\geq k}\mathbb{E}[||\mathbf{x}_{l}-\mathbf{1}_n\bar x_{l}||_R^2]$ converge to $\lim\sup_{k\rightarrow\infty} \mathbb{E}[||\bar x_{k}-x^{\star\top}||_2^2]$ and  $\lim\sup_{k\rightarrow\infty} \mathbb{E}[||\mathbf{x}_{k}-\mathbf{1}_n\bar x_{k}||_R^2]$, respectively, at the linear rate $\mathcal{O}(\rho(\mathbf{A})^k)$, where $\rho(\mathbf{A})<1$. In addition,
	\begin{equation}\label{eqf}
		\begin{aligned}
			&\limsup_{k\rightarrow\infty} \mathbb{E}[||\bar x_{k}-x^{\star\top}||_2^2]\leq [\mathbf{(I-A)^{-1}B}]_1,\\
			&\limsup_{k\rightarrow\infty} \mathbb{E}[||\mathbf{x}_{k}-\mathbf{1}_n\bar x_{k}||_R^2]\leq [\mathbf{(I-A)^{-1}B}]_2,\\
		\end{aligned}
	\end{equation}
	where $[\mathbf{(I-A)^{-1}B}]_i$ denotes the $i$th element of the vector $\mathbf{(I-A)^{-1}B}$. Their specific forms are given in \eqref{bound1} and \eqref{bound2}, respectively. 
\end{theorem} 
\begin{proof} In terms of Lemma \ref{mainl}, by induction we have
\begin{equation}\label{eqi}
\begin{aligned}
&\begin{bmatrix}
	\mathbb{E}[||\bar x_{k}-x^{\star\top}||_2^2|\mathcal{F}_k]\\
	\mathbb{E}[||\mathbf{x}_{k}-\mathbf{1}_n\bar x_{k}||_R^2|\mathcal{F}_k]\\
	\mathbb{E}[||\mathbf{\tilde y}_{k}-v\bar y_{k}]||_C^2|\mathcal{F}_k]
\end{bmatrix}\leq \mathbf{A}^k
\begin{bmatrix}
	\mathbb{E}[||\bar x_{0}-x^{\star\top}||_2^2]\\
	\mathbb{E}[||\mathbf{x}_{0}-\mathbf{1}_n\bar x_{0}||_R^2]\\
	\mathbb{E}[||\mathbf{\tilde y}_{0}-v\bar y_{0}]||_C^2]
\end{bmatrix}\\
&\qquad \qquad \qquad \qquad \qquad \qquad \quad \quad\quad \quad +\sum_{l=0}^{k-1}\mathbf{A}^l\mathbf{b}.
\end{aligned}
\end{equation}	
From equation \eqref{eqi}, we can see that if  $\rho(\mathbf{A})<1$, then $\sup_{l\geq k} \mathbb{E}[||\bar x_{l}-x^{\star\top}||_2^2]$, $\sup_{l\geq k}\mathbb{E}[||\mathbf{x}_{l}-\mathbf{1}_n\bar x_{l}||_R^2]$ and $\sup_{l\geq k}\mathbb{E}[||\mathbf{\tilde y}_{l}-v\bar y_{l}||_C^2]$ all converge to a neighborhood of $0$ at the linear rate $\mathcal{O}(\rho(\mathbf{A})^k)$. 


In view of Lemma \ref{lemma8}, it suffices to ensure $a_{11},a_{22}, a_{33}<1$ and det$(\mathbf{I-A})>0$, or
\begin{equation}\label{det}\begin{aligned}
	&\text{det}(\mathbf{I-A})=(1-a_{11})(1-a_{22})(1-a_{33})-a_{12}a_{23}a_{31}\\
	&-a_{13}a_{21}a_{32}-(1-a_{22})a_{13}a_{31}-(1-a_{11})a_{23}a_{32}\\
	&-(1-a_{33})a_{12}a_{21}>\frac{1}{2}(1-a_{11})(1-a_{22})(1-a_{33}),
	\end{aligned}\end{equation}
	which is equivalent to
	\begin{equation}\label{eqd}\begin{aligned}
		\frac{1}{2}(1-&a_{11})(1-a_{22})(1-a_{33})-c_1c_6c_8\eta^5\\&-c_2c_4c_9\eta^3-(1-a_{22})c_2c_8\eta^3-(1-a_{11})c_6c_9\eta^2\\
		&-(1-a_{33})c_1c_4\eta^3>0.
	\end{aligned} 	\end{equation}
	
	Next, we  give some sufficient conditions where $a_{11}, a_{22}, a_{33} < 1$ and relation \eqref{eqd} holds true.
	
	 First, $a_{11} < 1$ is ensured by choosing $ \eta'\leq 1/(\mu+L)$. In addition, $a_{22},a_{33}<1$ is ensured by choosing
	 \begin{equation}\label{eqr}
	 	1-a_{22}\geq\frac{1-\sigma_R^2}{3},\quad 1-a_{33}\geq\frac{1-\sigma_C^2}{3},
	 \end{equation}
	 requiring
	 \begin{equation}\label{eqeta}
	 	\eta \leq \min\Big\{\sqrt{\frac{1-\sigma_R^2}{6c_5}}, \sqrt{\frac{1-\sigma_C^2}{6c_{10}}}\Big\} .
	 \end{equation}
	 
	 Second, in view of relation \eqref{eqr}, $a_{22}>\frac{1+\sigma_R^2}{2},$ and	$a_{33}>\frac{1+\sigma_C^2}{2},$ we have
\begin{equation}\begin{aligned}
\frac{1}{2}(1-&a_{11})(1-a_{22})(1-a_{33})-c_1c_6c_8\eta^5-c_2c_4c_9\eta^3\\
&-(1-a_{22})c_2c_8\eta^3-(1-a_{11})c_6c_9\eta^2-(1-a_{33})c_1c_4\eta^3\\
		&>\frac{\eta'\mu}{2}\frac{1-\sigma_R^2}{3}\frac{1-\sigma_C^2}{3}-c_1c_6c_8\eta^5-c_2c_4c_9\eta^3-\eta'\mu c_6c_9\eta^2\\
		&-\frac{1-\sigma_R^2}{2}c_2c_8\eta^3-\frac{1-\sigma_C^2}{2}c_1c_4\eta^3.	\end{aligned}
	\end{equation}
	
	Then, relation \eqref{eqd} is equivalent to 
	$$d_1 \eta^4+ d_2\eta^2-d_3<0,$$
	where
	\begin{equation}\label{ds}
		\begin{aligned}
			&d_1:=c_1c_6c_8,\\
			&d_2:=c_2c_4c_9+\frac{1-\sigma_R^2}{2}c_2c_8+\frac{1-\sigma_C^2}{2}c_1c_4+\frac{u^\top\mathbf{T}v\mu c_6c_9}{n},\\
			&d_3:=\frac{u^\top\mathbf{T}v\mu}{18n}(1-\sigma_R^2)(1-\sigma_C^2).\\
		\end{aligned}
	\end{equation}
	
	Hence, a sufficient condition for det$(\mathbf{I-A})>0$ is
\begin{equation}\label{eqeta2}
	\eta\leq\sqrt{\frac{2d_3}{d_2+\sqrt{d_2^2+4d_1d_3}}}.
\end{equation}
Relation \eqref{eqeta} and  \eqref{eqeta2} yield the final bound on the stepsize $\eta$. 

Moreover, in light of \eqref{det} and \eqref{eqr}, we can obtain from \eqref{eqi} that
\begin{equation} \label{bound1}
	\begin{aligned}
		&[ \mathbf{(I-A)^{-1}B}]_1\\
		&=\Big[((1-a_{11})(1-a_{33})-a_{23}a_{32})c_3\eta^2\bar\theta^2\\
		&+(a_{13}a_{32}+a_{12}(1-a_{33}))c_7\eta^2\bar\theta^2\\
		&+(a_{12}a_{23}+a_{13}(1-a_{22}))c_{11}\bar\theta^2\Big]\frac{1}{\text{det}(\mathbf{I-A})}\\
		&\le \frac{18\bar\theta^2}{\eta'\mu(1-\sigma_R^2)(1- \sigma_C^2)}\Big[\Big(\frac{(1-\sigma_R^2)(1- \sigma_C^2)}{4}-c_6c_9\eta^2\Big)c_3\eta^2 \\
		&+\Big(c_2c_9+\frac{c_1(1- \sigma_C^2)}{2}\Big)c_7\eta^3
		+\Big(c_1c_6\eta^3+\frac{c_2\eta(1- \sigma_R^2)}{2}\Big)c_{11}\Big],\\
	\end{aligned}
\end{equation}
and
\begin{equation}\label{bound2}
	\begin{aligned}
		&[ \mathbf{(I-A)^{-1}B}]_2\\
&=\Big[(a_{23}a_{31}+a_{21}(1-a_{33}))c_3\eta^2\bar\theta^2\\
		&+(((1-a_{11})(1-a_{33})-a_{13}a_{31}))c_7\eta^2\bar\theta^2\\
		&+(a_{13}a_{21}+a_{23}(1-a_{11}))c_{11}\bar\theta^2\Big]\frac{1}{\text{det}(\mathbf{I-A})}\\
		&\le \frac{18\bar\theta^2}{\eta'\mu(1-\sigma_R^2)(1- \sigma_C^2)}\Big[\Big(c_6c_8\eta^4+\frac{c_4\eta^2(1- \sigma_C^2)}{2}\Big)c_3\eta^2 \\
		&+\Big(\frac{\eta'\mu(1- \sigma_C^2)}{2}-c_2c_8\eta^3\Big)c_7\eta^2
		+\Big(c_2c_4\eta^3+c_6\eta^2\eta'\mu\Big)c_{11}\Big].\\
	\end{aligned}
\end{equation}
	  \end{proof} 
\begin{remark}
	When $\eta$ is sufficiently small, it  can be shown that the linear rate indicator $\rho(\mathbf{A})\simeq 1-\eta'\mu$. From Theorem \ref{th1}, it is worth noting that the upper bounds in \eqref{bound1} and  \eqref{bound2} are functions of $\eta,\bar\theta $ and other problem parameters, and they are decreasing in terms of $\bar\theta$.  Fixing the system parameter and the privacy level $\min_i \epsilon_i$,  the optimization accuracy has the order of $d\sim O( 1/\min_i \epsilon_i)
$ for small $\epsilon_i$. As $\epsilon_i$ converges to 0, that is, for
complete privacy for each agents, the accuracy becomes arbitrarily low.

\end{remark}


\section{Differential Privacy Analysis}\label{privacy}
In this section, we analyze  the differential privacy property of SD-Push-Pull. The observation $\mathcal{O}_k$ denotes the message transmitted between agents, where {\color{black}$\mathcal{O}_k=\{x_{i,k}-\eta(y_{i,k+1}^\alpha-y_{i,k}^\alpha), c_{ji}y_{i,k}^\alpha| \forall i,j \in{\mathcal{N}}\}$}. {\color{black}Considering the differential privacy of agent $i_0$'s objective function $f_{i_0}$, we assume there exists a type of passive adversary called eavesdropper defined as follows, who is interested in the agent $i_0$.
\begin{definition}
	An eavesdropper is an external adversary who has access to all transmitted data $\mathcal{O}_k,\forall k$ by eavesdropping on the communications among the agents. Moreover, he  knows all other agents' functions  $\{f_j\}_{j\neq i_0}$, the network topology $\{\mathbf{R}, \mathbf{C}\}$ and the initial value of all agents $\{\mathbf{x}_0, \mathbf{y}_0\}$. 
	
	 \end{definition}}

	Before analyzing the differential privacy, we need the following assumption to bound the gradient of the objective function $f_i$ in the adjacency definition in Definition \ref{df1}.  
{\color{black}\begin{assumption}\label{asp3}
	Given a finite number of iterations $k\le K$, the gradients of all local objective functions $
	\nabla f_i(\mathbf{x_{i,k}}), \forall i\in \mathcal{N}, \forall k\leq K$, are bounded, i.e., there exists a positive constant $C$ such that for all $k\le K, ||\nabla f_i(\mathbf{x_{i,k}})||\le C, \forall i\in \mathcal{N}.$
\end{assumption}}

Next, we derive condition on the noise variance under which SD-Push-Pull satisfies $\epsilon$-differential privacy. 

{\color{black}\begin{theorem}
Given a finite number of iterations $k\le K$, under Assumption \ref{asp3}, SD-Push-Pull preserves $\epsilon_i$-differential privacy for a given $\epsilon_i>0$ of each agent $i$'s cost function if the noise parameter in \eqref{ya} satisfy $\theta_i >2\sqrt{p}CK/\epsilon_i$.

\end{theorem}}

\begin{proof}
 Since the eavesdropper is assumed to know the initial states of the algorithm, we have $\mathbf{x}_0^{(1)}=\mathbf{x}_0^{(2)}$ and $\mathbf{y}_0^{(1)}=\mathbf{y}_0^{(2)}$. 	{\color{black}From Algorithm \ref{alg1}, it  can be seen that given initial state $\{\mathbf{x}_0, \mathbf{y}_0\}$, the network topology $\{\mathbf{R}, \mathbf{C}\}$ and the function set $\mathcal{S}$, the observation sequence $\mathbf{z}=\{\mathcal{O}_k\}_{k=0}^K$ is uniquely determined by the noise sequence $\mathbf{\xi}=\{\xi_k\}_{k=0}^K$.  Hence, we use function $\mathcal{Z}_{\mathcal{F}}: (\mathbb{R}^{n\times p})^{K} \rightarrow (\mathbb{R}^{n\times p})^{K}$ to represent the relation, where $\mathcal{F}=\{\mathbf{x}_0, \mathbf{y}_0,\mathbf{R}, \mathbf{C},\mathcal{S}\}$, i.e., $\mathbf{z}=\mathcal{Z}_{\mathcal{F}}(\bf{\xi})$. From Definition \ref{df2}, keeping $\epsilon-$differential privacy is equivalent to guarantee that for any pair of adjacent function sets $\mathcal{S}^{(1)}$ and $\mathcal{S}^{(2)}$ and any observation $\mathcal{O}\subseteq \text{Range}(\mathcal{F})$,
$$P\{\mathbf{\xi}\in\Omega | \mathcal{Z}_{\mathcal{F}^{(1)}}(\mathbf{\xi}) \in \mathcal{O}\}\leq e^\epsilon P\{\mathbf{\xi}\in\Omega| \mathcal{Z}_{\mathcal{F}^{(2)}}(\mathbf{\xi}) \in \mathcal{O}\},
$$
where $\mathcal{F}^{(l)}=\{\mathbf{x}_0, \mathbf{y}_0,\mathbf{R}, \mathbf{C},\mathcal{S}^{(l)}\}, l=1,2$ and  $\Omega=(\mathbb{R}^{n\times p})^{K}$ denotes the sample space. 

Hence, we need to consider that $\mathcal{Z}_{\mathcal{F}^{(1)}}(\mathbf{\xi}^{(1)})=\mathcal{Z}_{\mathcal{F}^{(2)}}(\mathbf{\xi}^{(2)})$. It is indispensable to guarantee that  $\forall k \le K,  \forall i \in \mathcal{N}$, $y_{i,k+1}^{\alpha,(1)}=y_{i,k+1}^{\alpha,(2)}$ and $x_{i,k}^{(1)}=x_{i,k}^{(2)}$.  From the update rule of Algorithm \ref{alg1}, we can obtain that $y_{i,k+1}^{\alpha,(1)}=y_{i,k+1}^{\alpha,(2)}, \forall k\le K,  \forall i \in \mathcal{N}$ can lead to  $x_{i,k}^{(1)}=x_{i,k}^{(2)},\forall k\le K,  \forall i \in \mathcal{N}$. Thus, we only need to guarantee that $y_{i,k+1}^{\alpha,(1)}=y_{i,k+1}^{\alpha,(2)}, \forall k\le K,  \forall i \in \mathcal{N}$.

For any $i\neq i_0$,  to ensure  $y_{i,1}^{\alpha,(1)}=y_{i,1}^{\alpha,(2)}$, the noise should satify $\xi_{i_0}^{(1)}=\xi_{i_0}^{(2)}$ from \eqref{ya}. It then follows $y_{i,1}^{\beta,(1)}=y_{i,1}^{\beta,(2)}$ from \eqref{yb} since $f_i^{(1)}=f_i^{(2)}$. Based on the  above analysis, we can finally obtain
\begin{equation}\label{xi1}
	\xi_{i,k}^{(1)}=\xi _{i,k}^{(2)}, \quad \forall k\le K, \forall i\neq i_0.
\end{equation}

	Similarly, for agent $i_0$,  at iteration $k=0$, we have 
\begin{equation}\label{xi2}
	\xi_{i_0,0}^{(1)}=\xi_{i_0,0}^{(2)}
\end{equation}	
	
	For $k\ge1$, since $f_{i_0}^{(1)}\ne f_{i_0}^{(2)}$, the noise should satisfy 
	\begin{equation}\label{es}
		\Delta \xi_{i_0,k}=-(1-\beta_{i_0}) \Delta y_{i_0,k}^\beta, \forall 1 \le  k\le K,
	\end{equation}
	where $\Delta \xi_{i_0,k}=\xi_{i_0,k}^{(1)}-\xi_{i_0,k}^{(2)}$ and $\Delta y_{i_0,k}^\beta=y_{i_0,k}^{\beta,(1)}-y_{i_0,k}^{\beta,(2)}$.
	In light of equation \eqref{yb}, we have 
\begin{equation}
	\Delta y_{i_0,k}^\beta=\beta_{i_0}\Delta y_{i_0,k-1}^\beta+\Delta f_{i_0,k-1}, \forall k\ge1,
\end{equation}
where	$\Delta f_{i_0,k-1}=\nabla f_{i_0}^{(1)}(x_{i_0,k}^{(1)})-\nabla f_{i_0}^{(2)}(x_{i_0,k}^{(2)})$. 

Given $\Delta y_{i_0,0}^\beta=0$, by induction we have the following relationship, for any $k\ge 1$,
\begin{equation}\label{yite}
	\Delta y_{i_0,k}^\beta=\beta_{i_0}^{k-1}\Delta f_{i_0,0}+\beta_{i_0}^{k-2}\Delta f_{i_0,1}+\cdots+\beta_{i_0}\Delta f_{i_0,k-1}.
\end{equation}
 Combining \eqref{es} and \eqref{yite}, we have
 \begin{equation}
 	\begin{aligned}
 		\sum\limits_{k=0}^K ||\Delta \xi_{i_0,k}||_1&\le (1-\beta_{i_0}) \sum\limits_{k=0}^K || \Delta y_{i_0,k}^\beta||_1\\
 		&\leq (1-\beta_{i_0})\sum\limits_{k=0}^K \sum\limits_{t=0}^{k-1} \beta_{i_0}^{k-1-t}||\Delta f_{i_0,t}||_1\\
 		&\le (1-\beta_{i_0})\sum\limits_{k=0}^K \sum\limits_{t=0}^{k-1} \beta_{i_0}^{k-1-t}2\sqrt{p}C\\
 		&< (1-\beta_{i_0})\sum\limits_{k=1}^K (\beta_{i_0}^{k-1})2K\sqrt{p}C\\\\
 		&= \frac{2(1-\beta_{i_0})\sqrt{p}CK(1-\beta_{i_0}^K)}{1-\beta_{i_0}}\\
 		&<2\sqrt{p}CK,
 	\end{aligned}
 \end{equation}
 where the second inequality holds based on Assumption \ref{asp3}.

	Denote $\mathcal{R}^{(l)}=\{\mathbf{\xi}^{(l)}|\mathcal{Z}_{\mathcal{F}^{(l)}}(\bf{\xi}^{(l)}) \in \mathcal{O}\}$, $l=1,2,$ then
	\begin{equation}
	\begin{aligned}
				P\{\mathbf{\xi}\in\Omega | &\mathcal{Z}_{\mathcal{F}^{(l)}}(\mathbf{\xi}) \in \mathcal{O}\}\\
				&=P\{\bf{\xi}^{(l)}\in \mathcal{R}^{(l)}\}	=\int_{\mathcal{R}^{(l)}} f_{\xi}(\bf{\xi}^{(l)})d\mathbf{\xi}^{(l)},
	\end{aligned}
		\end{equation}
	where $f_{\xi}(\mathbf{\xi}^{(l)})=\prod_{k=0}^K\prod_{i=1}^n\prod_{r=1}^p f_{L}([\xi_{i,k}^{(l)}]_r,\theta_i)$.
	
	According to the above relation \eqref{xi1}-\eqref{es}, we can obtain for any $\mathbf{\xi}^{(1)}$, there exists a  $\mathbf{\xi}^{(2)}$ such that $\mathcal{Z}_{\mathcal{F}^{(1)}}(\mathbf{\xi}^{(1)})=\mathcal{Z}_{\mathcal{F}^{(2)}}(\mathbf{\xi}^{(2)})$. As the converse argument is also true, the above defines a bijection.  Hence, for any $\mathbf{\xi}^{(2)}$, there exists a unique $(\mathbf{\xi}^{(1)}, \Delta\mathbf{\xi})$	such that $\mathbf{\xi}^{(2)}=\mathbf{\xi}^{(1)}+\Delta\mathbf{\xi}$. Since $\Delta\mathbf{\xi}$ is fixed and is not dependent on $\mathbf{\xi}^{(2)}$, we can use a change of variables to obtain 
	\begin{equation}
		P\{\bf{\xi}^{(2)}\in \mathcal{R}^{(2)}\}	=\int_{\mathcal{R}^{(1)}} f_{\xi}(\mathbf{\xi}^{(1)}+\Delta \mathbf{\xi})d\mathbf{\xi}^{(1)}.
	\end{equation}
	
	Hence, we have
	\begin{equation}
		\begin{aligned}
			&\frac{P\{\mathbf{\xi}\in\Omega |\mathcal{Z}_{\mathcal{F}^{(1)}}(\mathbf{\xi}) \in \mathcal{O}\}}{P\{\mathbf{\xi}\in\Omega |\mathcal{Z}_{\mathcal{F}^{(2)}}(\mathbf{\xi}) \in \mathcal{O}\}}=\frac{P\{\bf{\xi}^{(1)}\in \mathcal{R}^{(1)}\}}{P\{\bf{\xi}^{(2)}\in \mathcal{R}^{(2)}\}}\\
			&=\frac{\int_{\mathcal{R}^{(1)}} f_{\xi}(\mathbf{\xi}^{(1)})d\mathbf{\xi}^{(1)}}{\int_{\mathcal{R}^{(1)}} f_{\xi}(\mathbf{\xi}^{(1)}+\Delta \mathbf{\xi})d\mathbf{\xi}^{(1)}}.\\
		\end{aligned}
	\end{equation}
	 Since 
	 \begin{equation}
	 	\begin{aligned}
	 		&\frac{ f_{\xi}(\mathbf{\xi}^{(1)})d\mathbf{\xi}^{(1)}}{\ f_{\xi}(\mathbf{\xi}^{(1)}+\Delta \mathbf{\xi})d\mathbf{\xi}^{(1)}}\\
	 		&=\prod_{k=0}^K\prod_{i=1}^n\prod_{r=1}^p \frac{f_{L}([\xi_{i,k}^{(1)}]_r,\theta_i)}{f_{L}([\xi_{i,k}^{(1)}+\Delta \xi_{i,k}]_r,\theta_i)}\\
	 		&\leq \prod\limits_{k=0}^K\prod\limits_{r=1}^p\exp\Big(
		\frac{\Big|[\Delta \xi_{i_0,k}]_r\Big|}{\theta_{i_0}}\Big)\\
		&=\exp\Big(\sum\limits_{k=0}^K\frac{||\Delta \xi_{i_0,k}||_1}{\theta_{i_0}}\Big)\leq e^{\epsilon_{i_0}}.
	 	\end{aligned}
	 \end{equation}
Then, integrating both sides over $\mathcal{R}^{(1)}$, we have 
\begin{equation}
	\begin{aligned}
			&P\{\mathbf{\xi}\in\Omega |\mathcal{Z}_{\mathcal{F}^{(1)}}(\mathbf{\xi}) \in \mathcal{O}\}\le e^{\epsilon_{i_0}} P\{\mathbf{\xi}\in\Omega |\mathcal{Z}_{\mathcal{F}^{(2)}}(\mathbf{\xi}),\\
		\end{aligned}
\end{equation}	 
which establishes the $\epsilon_{i_0}$-differential privacy  of agent $i_{0}$. The fact that $i_0$ can be arbitrary agent without loss of generality guarantees $\epsilon_i$-differential privacy of each agent $i$.

}

\end{proof}

\section{SIMULATIONS}
In this section, we illustrate the effectiveness of  SD-Push-Pull.

Consider a network containing $N=5$ agents, shown in Fig. \ref{digraph}. The optimization problem is considered as the ridge regression problem, i.e.,
\begin{equation}
	\min\limits_{x\in\mathbb{R}^p} f(x) =\frac{1}{n}\sum\limits_{i=1}^n f_i(x)=\frac{1}{5}\sum\limits_{i=1}^n\Big((u_i^\top x-v_i)^2+\rho ||x||^2_2\Big)\end{equation}
	where $\rho>0$ is a penalty parameter. Each agent $i$ has its private sample $(u_i,v_i)$ where $u_i\in\mathbb{R}^p$ denotes the features and  $v_i\in\mathbb{R}$ denotes the observed outputs. The vector  $u_i\in[-1,1]^p$ is drawn from the uniform distribution. Then the observed outputs $v_i$ is generated according to $v_i=u_i^\top \tilde x_i +\gamma_i$, where  $\tilde x_i $ is evenly located in $[0,10]^p$ and $\gamma_i\sim \mathcal{N}(0,5)$. In terms of the above parameters, problem \eqref{pro1} has a unique solution $x^{\star\top}=(\sum_{i=1}^n[u_i u_i^\top+n\rho \mathbf{I}])^{-1}\sum_{i=1}^n[u_i u_i^\top]\tilde x_i$.
\begin{figure}[htp]
\centering
 \includegraphics[width=0.20\textwidth]{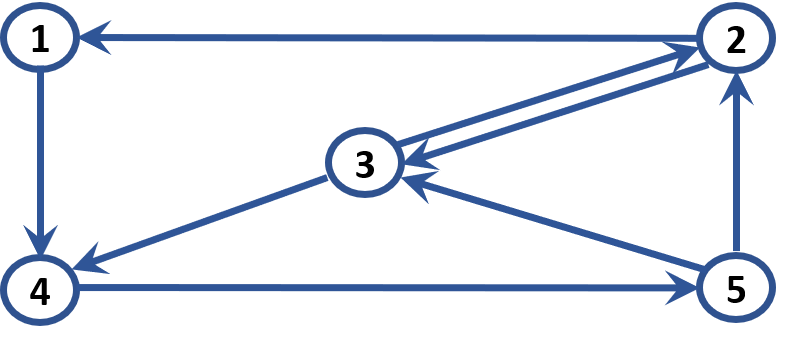}
  \caption{A digraph of 5 agents. }
      \label{digraph}
\end{figure}

The weight between two substates, $\alpha_i$ and $\beta_i$, are set to be $0.01$ and $0.5$ for each agent $i\in\{1,2,3,4,5\}$. The matrix $\mathbf{R}$ and $\mathbf{C}$ are designed as follows: for any agent $i$, $R_{ij}=\frac{1}{|\mathcal{N}_{\mathbf{R},i}^{in}|+1}$ for $j\in \mathcal{N}_{\mathbf{R},i}^{in}$ and $R_{ii}=1-\sum_{j\in \mathcal{N}_{\mathbf{R},i}^{in}} R_{ij}$; for any agent $i$, $C_{li}=\frac{1-\alpha_i}{|\mathcal{N}_{\mathbf{C},i}^{out}|+1}$ for all $l\in \mathcal{N}_{\mathbf{C},i}^{out}$ and $C_{ii}=1-\alpha_i-\sum_{l\in \mathcal{N}_{\mathbf{C},i}^{out}} C_{li}$.

Assume  $\epsilon_i=\epsilon, \forall  i=\{1,2,3,4,5\}$ and $\delta=10$. To investigate the dependence of the algorithm accuracy with differential privacy level, we compare the performance SD-Push-Pull  for three cases:
$\epsilon=1, \epsilon=5$ and $\epsilon=10$, in terms of the normalized residual $\frac{1}{5}\mathbb{E}\Big[\sum\limits_{i=1}^5\frac{||x_{i,k}-x^{\star\top}||_2^2}{||x_{i,0}-x^{\star\top}||_2^2}\Big]$. The results are depicted in Fig. \ref{per}, which reflect that SD-Push-Pull becomes suboptimal to guarantee differential privacy, and the constant $\epsilon$ determines a tradeoff between the privacy level and the optimization accuracy.

\begin{figure}[htp]
\centering
 \includegraphics[width=0.4\textwidth]{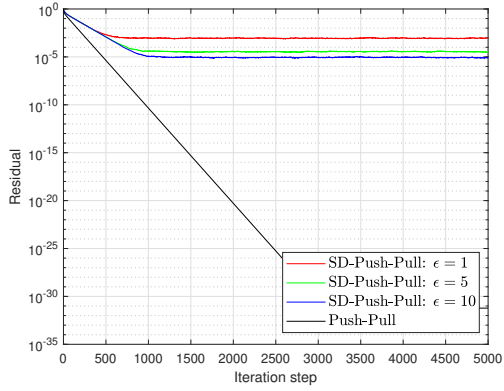}
  \caption{Evolutions of the normalized residual  under different settings of the privacy level. The expected
residual are approximated by averaging over 50 simulation results. Dimension $ p= 10$, stepsize $\alpha = 0.01$ and penalty parameter $\rho = 0.01$.}
      \label{per}
\end{figure}

\section {CONCLUSION AND FUTURE WORK}

In this paper, we  considered a distributed optimization problem with differential privacy in the scenario where a network is abstracted as an unbalanced directed graph. We proposed a state-decomposition-based differentially private distributed optimization algorithm (SD-Push-Pull). In particular, the state decomposition mechanism was adopted to guarantee the differential privacy of individuals’ sensitive information. In addition, we proved that each agent reach a neighborhood of the optimum in expectation exponentially fast under a constant stepsize policy. Moreover, we showed that the constants $(\epsilon,\delta)$ determine a tradeoff between the privacy level and the optimization accuracy. Finally, a numerical example was provided that demonstrates the effectiveness of SD-Push-Pull. Future work includes improving the accuracy of the optimization and considering the optimization problem with  constraints.

\section{APPENDIX}
\subsection{Proof of Lemma \ref{mainl}}\label{ap1}
The three inequalities embedded in \eqref{ine} come from \eqref{m1}, \eqref{m2} and \eqref{m3}, repectively. 

\textit{First inequality}: By Lemma \ref{lemma2}, Lemma \ref{lemma3}, Lemma \ref{lemma6} and Lemma \ref{lm8}, we can obtain from \eqref{m1} that
\begin{equation*}
	\begin{aligned}
		&\mathbb{E}[||\bar x_{k+1}-x^{\star\top}||_2^2|\mathcal{F}_k]\\
		&= \mathbb{E}[||\bar x_{k}-\eta'g_k-x^{\star\top}-\eta'(h_k-g_k)-\frac{\eta}{n}u^\top\mathbf{T}(\mathbf{\tilde y}_k- v\bar y_k)||_2^2|\mathcal{F}_k]\\
		&+\mathbb{E}[||\frac{\eta'}{n}\mathbf{1}_{n}^\top \mathbf{\xi}_k||_2^2|\mathcal{F}_k],\\
		&\leq \tau_1 \mathbb{E}[||\bar x_{k}-\eta'g_k-x^{\star\top}||_2^2|\mathcal{F}_k]+\frac{2\tau_1}{\tau_1-1}\mathbb{E}[||\eta'(h_k-g_k)||_2^2|\mathcal{F}_k]\\
		&+\frac{2\tau_1}{\tau_1-1}\mathbb{E}[||\frac{\eta}{n}u^\top\mathbf{T}(\mathbf{\tilde y}_k- v\bar y_k)||_2^2|\mathcal{F}_k]+\frac{2p\eta'^2}{n}\bar\theta^2.\\
				&\leq \tau_1(1-\eta'\mu)^2 \mathbb{E}[||\bar x_{k}-x^{\star\top}||_2^2|\mathcal{F}_k]+\frac{2\tau_1}{\tau_1-1}\mathbb{E}[||\eta'(h_k-g_k)||_2^2|\mathcal{F}_k]\\
		&+\frac{2\tau_1}{\tau_1-1}\mathbb{E}[||\frac{\eta}{n}u^\top\mathbf{T}(\mathbf{\tilde y}_k- v\bar y_k)||_2^2|\mathcal{F}_k]+\frac{2p\eta'^2}{n}\bar\theta^2.\\
			\end{aligned}
\end{equation*}

Taking $\tau_1=\frac{1}{1-\eta'\mu}$, we have
\begin{equation*}
	\begin{aligned}
	&\mathbb{E}[||\bar x_{k+1}-x^{\star\top}||_2^2|\mathcal{F}_k]\\
		&\leq (1-\eta' \mu)\mathbb{E}[||\bar x_{k}-x^{\star\top}||_2^2|\mathcal{F}_k]+c_1\eta\mathbb{E}[||\mathbf{x}_k-\mathbf{1}_n{\bar x}_k||_2^2|\mathcal{F}_k]\\&+c_2\eta\mathbb{E}[||\mathbf{y}_k- v\bar y_k||_2^2|\mathcal{F}_k]+c_3\eta^2\bar\theta^2,
	\end{aligned}
\end{equation*}
 where 
 \begin{equation}\label{c1}
 	c_1=\frac{2u^\top\mathbf{T}vL^2}{\mu n^2},\quad c_2=\frac{2||u^\top\mathbf{T}||_2^2}{u^\top\mathbf{T}v\mu n}, \quad c_3=\frac{2p(u^\top\mathbf{T}v)^2}{n^3}.
 \end{equation}


\textit{Second inequality:} By relation \eqref{m2}, Lemma \ref{lemma5}, Lemma \ref{lemma6} and Lemma \ref{lm8}, we can obtain
\begin{equation*}
	\begin{aligned}
		&\mathbb{E}[||\mathbf{x}_{k+1}-\mathbf{1}_n\bar x_{k+1}||^2_R|\mathcal{F}_k]\\
		&\leq \tau_2\sigma_R^2 \mathbb{E}[||\mathbf{x}_{k}-\mathbf{1}_n\bar x_{k}||_R^2|\mathcal{F}_k]+\frac{\tau_2\sigma_R^2\eta^2}{\tau_2-1}||\mathbf{T}||_R^2\mathbb{E}[||\mathbf{\tilde y}_k||_R^2|\mathcal{F}_k]\\
		&\leq \tau_2\sigma_R^2 \mathbb{E}[||\mathbf{x}_{k}-\mathbf{1}_n\bar x_{k}||_R^2|\mathcal{F}_k]+\frac{2\tau_2\sigma_R^2\eta^2}{\tau_2-1}||\mathbf{T}v||_R^2\mathbb{E}[||\bar y_k||_R^2|\mathcal{F}_k]\\
		&+\frac{2\tau_2\sigma_R^2\eta^2}{\tau_2-1}||\mathbf{T}||_R^2\mathbb{E}[||\mathbf{\tilde  y}_k-v\bar y_k||_R^2|\mathcal{F}_k].\\
			\end{aligned}
\end{equation*}
In view of Lemma \ref{lemma2},
\begin{equation}\label{bary}
\begin{aligned}
	\mathbb{E}[||\bar y_k||_2^2|\mathcal{F}_k]\leq \frac{2p\bar \theta^2}{n}&+ \frac{2L^2}{n}\mathbb{E}[||\mathbf{x}_k-\mathbf{1}_n \bar x_k||_2^2|\mathcal{F}_k]\\
	&+2L^2\mathbb{E}[||\bar x_k-x^{\star\top}||_2^2|\mathcal{F}_k].
	\end{aligned}
\end{equation}

Taking $\tau_2=\frac{1+\sigma_R^2}{2\sigma_R^2}$, given $1+\sigma_R^2<2$, we  can obtain
\begin{equation*}
	\begin{aligned}
		&\mathbb{E}[||\mathbf{x}_{k+1}-\mathbf{1}_n\bar x_{k+1}||^2_R|\mathcal{F}_k] \leq \frac{1+\sigma_R^2}{2}\mathbb{E}[||\mathbf{x}_{k}-\mathbf{1}_n\bar x_{k}||_R^2|\mathcal{F}_k]\\
		&+\frac{4\sigma_R^2\eta^2\delta_{R,2}||\mathbf{T}v||_R^2}{1-\sigma_R^2}\mathbb{E}[||\bar y_k||_2^2|\mathcal{F}_k]\\
		&+\frac{4\sigma_R^2\eta^2\delta_{R,C}}{1-\sigma_R^2}||\mathbf{T}||_R^2\mathbb{E}[||\mathbf{\tilde  y}_k-v\bar y_k||_C^2|\mathcal{F}_k]\\
		&\le c_4\eta^2\mathbb{E}[||\bar x_k-x^{\star\top}||_2^2|\mathcal{F}_k]+(\frac{1+\sigma_R^2}{2}+c_5\eta^2)\\
		&\times\mathbb{E}[||\mathbf{x}_{k}-\mathbf{1}_n\bar x_{k}||_R^2|\mathcal{F}_k]+c_6\eta^2\mathbb{E}[||\mathbf{\tilde  y}_k-v\bar y_k||_C^2|\mathcal{F}_k]+c_7\eta^2\bar\theta^2,
			\end{aligned}
\end{equation*}
where 
\begin{equation}\label{c2}
	\begin{aligned}
		&c_4=\frac{8\sigma_R^2L^2\delta_{R,2}||\mathbf{T}v||_R^2}{1-\sigma_R^2}, \quad c_5=\frac{8\sigma_R^2L^2\delta_{R,2}||\mathbf{T}v||_R^2}{(1-\sigma_R^2)n},\\
		&c_6=\frac{4\sigma_R^2\delta_{R,C}||\mathbf{T}||_R^2}{1-\sigma_R^2},\qquad c_7=\frac{8p\sigma_R^2\delta_{R,2}||\mathbf{T}v||_R^2}{(1-\sigma_R^2)n}.\\
	\end{aligned}
\end{equation}


\textit{Third inequality:} It follows from \eqref{m3}, Lemma \ref{lemma5}, Lemma \ref{lemma6} and Lemma \ref{lm8} that
\begin{equation*}
	\begin{aligned}
		&\mathbb{E}[||\mathbf{\tilde y}_{k+1}-v\bar y_{k+1}||_C^2|\mathcal{F}_k]\\
		&\leq \frac{1+\sigma_C^2}{2}\mathbb{E}[ ||\mathbf{y}_{k}-v\bar y_{k}||_C^2|\mathcal{F}_k]\\
		&+\frac{(1+\sigma_C^2)\delta_{C,2}^2}{(1-\sigma_C^2)}||\mathbf{I}_{2n}-\frac{v\mathbf{1}_{2n}^\top}{n}||_2^2\mathbb{E}[||\nabla \tilde F(\mathbf{x}_{k+1})-\nabla \tilde F(\mathbf{x}_{k})||_2^2|\mathcal{F}_k].\\
	\end{aligned}
\end{equation*}

Next, we bound $\mathbb{E}[||\nabla \tilde F(\mathbf{x}_{k+1})-\nabla \tilde F(\mathbf{x}_{k})||_2^2|\mathcal{F}_k]$.

\begin{equation*}
	\begin{aligned}
	&\mathbb{E}[||\nabla \tilde F(\mathbf{x}_{k+1})-\nabla \tilde F(\mathbf{x}_{k})||_2^2|\mathcal{F}_k]\\
	&=\mathbb{E}[||[(\mathbf{\xi}_{k+1}-\mathbf{\xi}_k)^\top, (\nabla F(\mathbf{x}_{k+1})-\nabla F(\mathbf{x}_k))^\top]^\top||_2^2|\mathcal{F}_k]\\
		&=\mathbb{E}[||\mathbf{\xi}_{k+1}-\mathbf{\xi}_k||_2^2|\mathcal{F}_k]+\mathbb{E}[||\nabla F(\mathbf{x}_{k+1})-\nabla F(\mathbf{x}_k)||_2^2|\mathcal{F}_k]\\
		&\leq4pn\bar\theta^2+\mathbb{E}[||\nabla F(\mathbf{x}_{k+1})-\nabla F(\mathbf{x}_k)||_2^2|\mathcal{F}_k]\\
			\end{aligned}
\end{equation*}
Then, From Assumption \ref{asp} and equation \eqref{bary}, we have
\begin{equation*}
	\begin{aligned}
		&\mathbb{E}[||\nabla F(\mathbf{x}_{k+1})-\nabla F(\mathbf{x}_k)||_2^2|\mathcal{F}_k]\\
		&\leq L^2\mathbb{E}[||\mathbf{x}_{k+1}-\mathbf{x}_{k}||_2^2|\mathcal{F}_k]\\
		&\leq L^2\mathbb{E}[||(\mathbf{R}-\mathbf{I})(\mathbf{x}_{k}-\mathbf{1}_n\bar x_k)\\
		&\qquad \qquad \qquad \qquad \qquad-\eta\mathbf{RT}(\mathbf{\tilde y}_k-v\bar y_k)-\eta\mathbf{RT}v\bar y_k||_2^2|\mathcal{F}_k]\\
				&\le 3L^2||\mathbf{R}-\mathbf{I}||_2^2\mathbb{E}[||\mathbf{x}_{k}-\mathbf{1}_n\bar x_k||_2^2|\mathcal{F}_k]+3L^2\eta^2||\mathbf{R}\mathbf{T}||_2^2\\
				&\times \mathbb{E}[||\mathbf{\tilde y}_k-v\bar y_k||_2^2|\mathcal{F}_k]+3L^2\eta^2||\mathbf{R}\mathbf{T}||_2^2||v||_2^2\mathbb{E}[||\bar y_k||_2^2|\mathcal{F}_k]
	\end{aligned}
\end{equation*}

In light of inequality \eqref{bary}, we further have

\begin{equation*}
	\begin{aligned}
		&\mathbb{E}[||\nabla F(\mathbf{x}_{k+1})-\nabla F(\mathbf{x}_k)||_2^2|\mathcal{F}_k]\\
				&\le (3L^2||\mathbf{R}-\mathbf{I}||_2^2+\frac{6L^4\eta^2||\mathbf{R}\mathbf{T}||_2^2||v||_2^2}{n})\mathbb{E}[||\mathbf{x}_{k}-\mathbf{1}_n\bar x_k||_2^2|\mathcal{F}_k]\\
				&+3L^2\eta^2||\mathbf{R}\mathbf{T}||_2^2\mathbb{E}[||\mathbf{\tilde y}_k-v\bar y_k||_2^2|\mathcal{F}_k]+\frac{6L^2\eta^2||\mathbf{R}\mathbf{T}||_2^2||v||_2^2p\bar\theta^2}{n}\\
				&+6L^4\eta^2||\mathbf{R}\mathbf{T}||_2^2||v||_2^2\mathbb{E}[||\bar x_k-x^{\star\top}||_2^2|\mathcal{F}_k]
	\end{aligned}
\end{equation*}

Thus, given that $\eta< \frac{1}{L}$, we can obtain

\begin{equation*}
	\begin{aligned}
		&\mathbb{E}[||\mathbf{\tilde y}_{k+1}-v\bar y_{k+1}||_C^2|\mathcal{F}_k]\\
		&\le c_8\eta^2\mathbb{E}[||\bar x_k-x^{\star\top}||_2^2|\mathcal{F}_k]+c_9\mathbb{E}[||\mathbf{x}_{k}-\mathbf{1}_n\bar x_k||_2^2|\mathcal{F}_k]\\
		&+(\frac{1+\sigma_C^2}{2}+c_{10}\eta^2)\mathbb{E}[||\mathbf{\tilde y}_k-v\bar y_k||_2^2|\mathcal{F}_k]+c_{11}\bar\theta^2,
	\end{aligned}
\end{equation*}
where
\begin{equation}\label{c3}
	\begin{aligned}
		&c_8=\frac{12\delta_{C,2}^2L^4}{(1-\sigma_C^2)}||\mathbf{I}_{2n}-\frac{v\mathbf{1}_{2n}^\top}{n}||_2^2||\mathbf{R}\mathbf{T}||_2^2||v||_2^2,\\
		&c_9=\frac{2\delta_{C,2}^2}{(1-\sigma_C^2)}||\mathbf{I}_{2n}-\frac{v\mathbf{1}_{2n}^\top}{n}||_2^2\Big(3L^2||\mathbf{R}-\mathbf{I}||_2^2\\
		&+\frac{6L^2||\mathbf{R}\mathbf{T}||_2^2||v||_2^2}{n}\Big),\\
		&c_{10}=\frac{6\delta_{C,2}^2L^2}{(1-\sigma_C^2)}||\mathbf{I}_{2n}-\frac{v\mathbf{1}_{2n}^\top}{n}||_2^2||\mathbf{R}\mathbf{T}||_2^2,\\
		&c_{11}=\frac{12\delta_{C,2}^2p}{(1-\sigma_C^2)n}||\mathbf{I}_{2n}-\frac{v\mathbf{1}_{2n}^\top}{n}||_2^2||\mathbf{R}\mathbf{T}||_2^2||v||_2^2.\\
	\end{aligned}
\end{equation}


\bibliographystyle{unsrt}
\bibliography{ref.bib}
\begin{IEEEbiography}[{\includegraphics[width=1in,height=1.25in,clip,keepaspectratio]{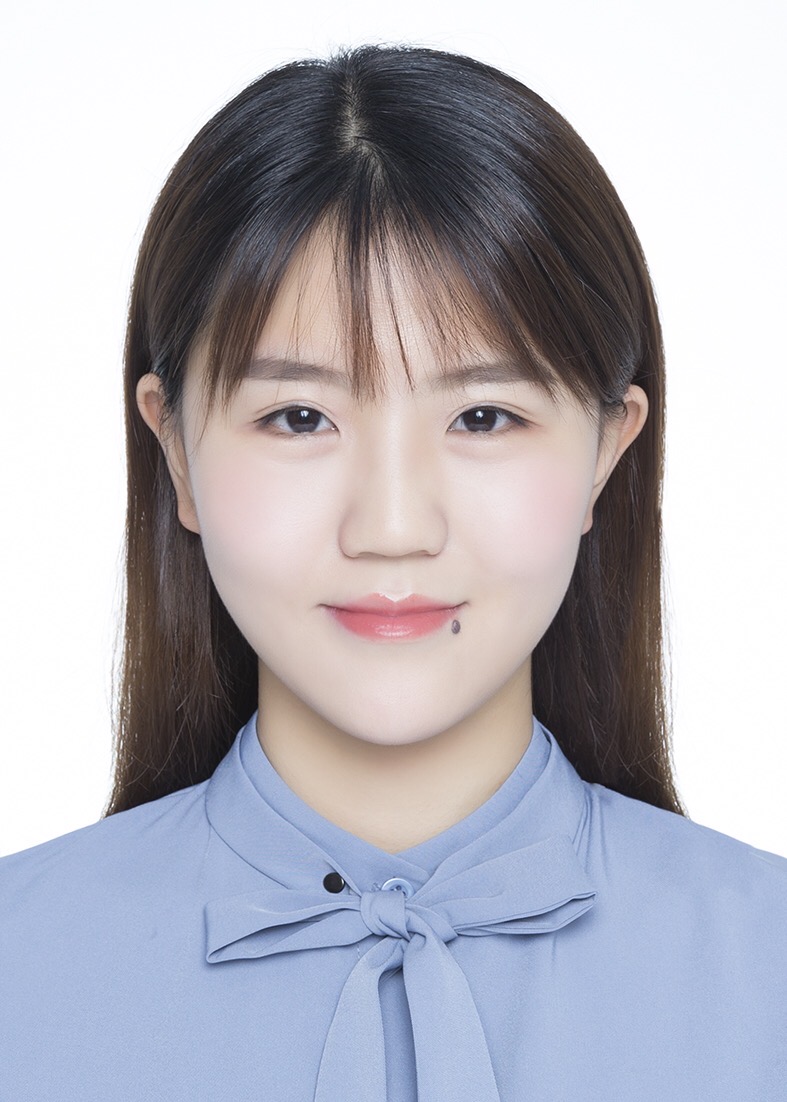}}]{Xiaomeng Chen} received her B.S. degree in Electrical Science and Engineering from Nanjing University, JiangSu, China, in 2019. She is currently pursuing   the Ph.D degree in Electrical and Computer Engineering from Hong Kong University of Science and Technology, Hong Kong. Her current research interests include  cyber-physical system security/privacy, compressed communication, event-triggered mechanism and distributed optimization.
\end{IEEEbiography}
\begin{IEEEbiography}[{\includegraphics[width=1in,height=1.25in,clip,keepaspectratio]{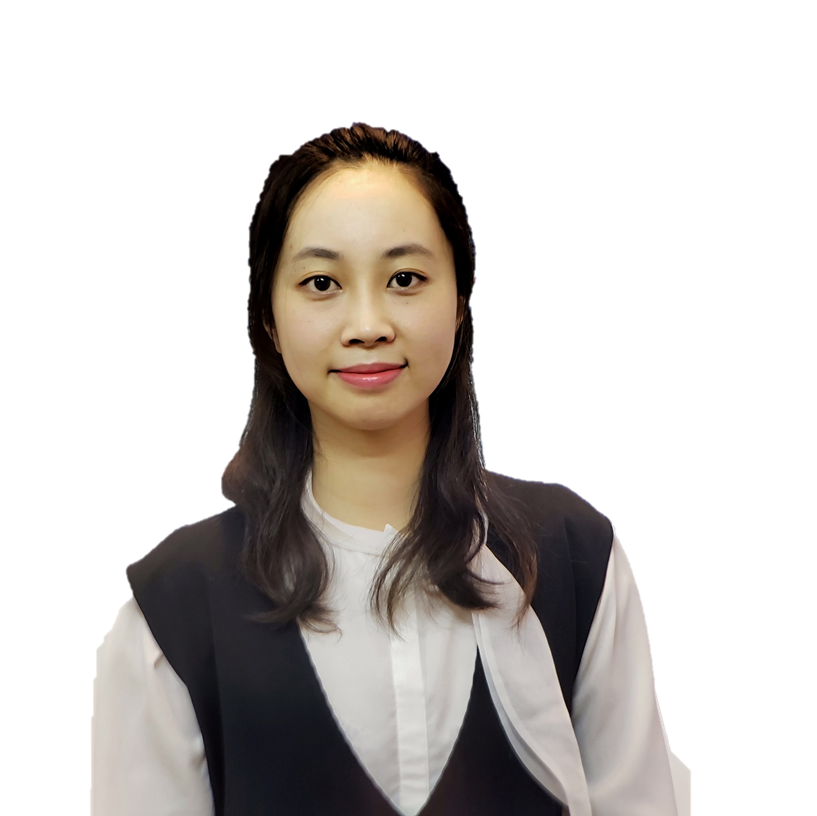}}]{Lingying Huang} received her B.S. degree in Electrical Engineering and Automation from Southeast University, JiangSu, China, in 2017, and the Ph.D degree in Electrical and Computer Engineering from Hong Kong University of Science and Technology, Hong Kong, in 2021. She is currently a Research fellow at the School of Electrical and Electronic Engineering,  Nanyang Technological University. From July 2015 to August 2015, she had a summer program in Georgia Tech Univerisity, USA.  Her current research interests include intelligent vehicles, cyber-physical system security/privacy, networked state estimation, event-triggered mechanism and distributed optimization.
\end{IEEEbiography}
\begin{IEEEbiography}[{\includegraphics[width=1in,height=1.25in,clip,keepaspectratio]{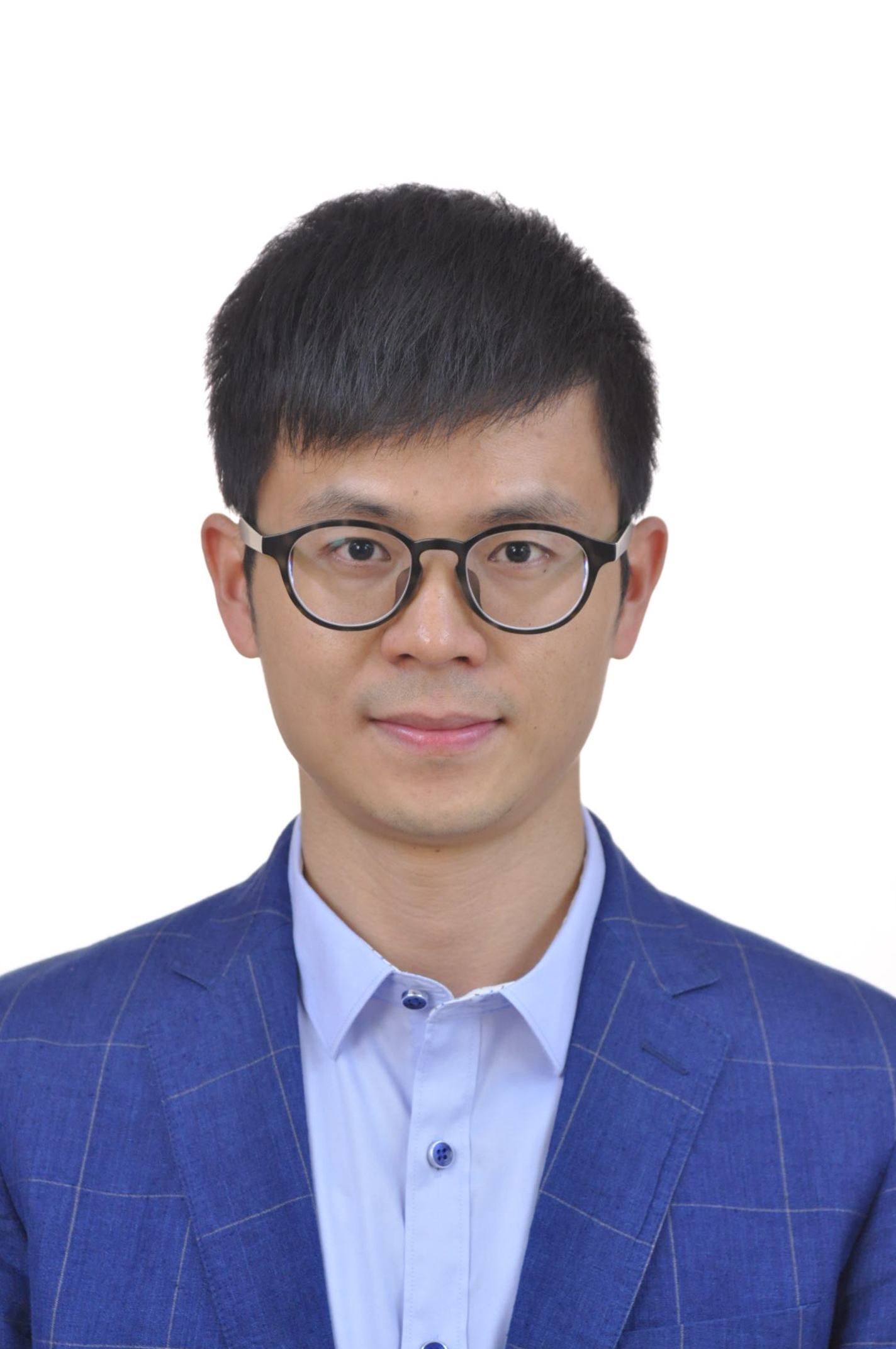}}]{Lidong He} received the B. Eng. Degree in mechanical engineering from Zhejiang Ocean University, Zhoushan, China, in 2005. He received the Master’s degree from Northeastern University, Shenyang, China, in 2008 and the Ph.D. degree in control science and engineering from Shanghai Jiao Tong University, Shanghai, China, in 2014. In the fall of 2010 and 2011, he was a visiting student with The Hong Kong University of Science and Technology.
 
From 2014 to 2016, He was a postdoctoral researcher with Zhejiang University, Hangzhou, China. In 2016, He joined the School of Automation, Nanjing University of Science and Technology and now is an Associate professor.
 
His research interests include distributed control of multi-agent systems, secure estimation and control for cyber physical systems. He is an active reviewer for many international journals.
\end{IEEEbiography}

\begin{IEEEbiography}
[{\includegraphics[width=1in,height=1.25in,clip,keepaspectratio]{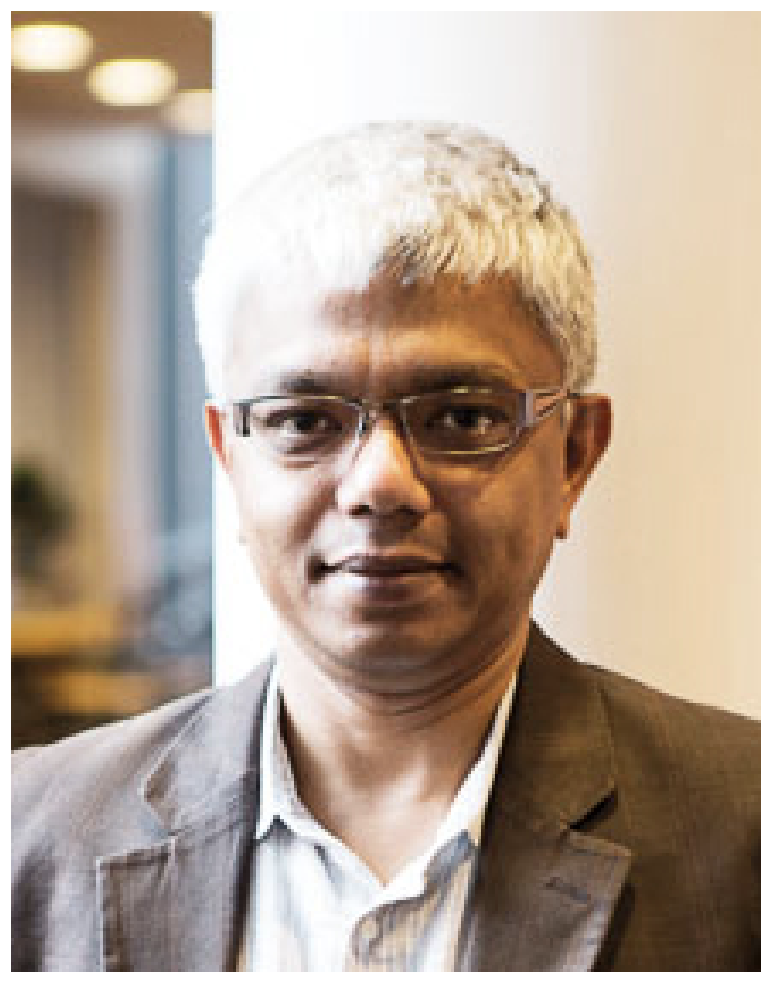}}]
	{Subhrakanti Dey} received the Bachelor in Technology and Master in Technology degrees from the Department of Electronics and Electrical Communication Engineering, Indian Institute of Technology, Kharagpur, in 1991 and 1993, respectively, and the Ph.D. degree from the Department of Systems Engineering, Research School of Information Sciences and Engineering, Australian National University, Canberra, in 1996.

He is currently a Professor with the Hamilton Institute, National University of Ireland, Maynooth, Ireland. Prior to this, he was a Professor with the Dept. of Engineering Sciences in Uppsala University, Sweden (2013-2017), Professor with the Department of Electrical and Electronic Engineering, University of Melbourne, Parkville, Australia, from 2000 until early 2013, and a Professor of Telecommunications at University of South Australia during 2017-2018.  From September 1995 to September 1997, and September 1998 to February 2000, he was a Postdoctoral Research Fellow with the Department of Systems Engineering, Australian National University. From September 1997 to September 1998, he was a Postdoctoral Research Associate with the Institute for Systems Research, University of Maryland, College Park.

His current research interests include wireless communications and networks, signal processing for sensor networks, networked control systems, and molecular communication systems.

Professor Dey currently serves as a Senior Editor on the Editorial Board IEEE Transactions on Control of Network Systems, and as an Associate Editor/Editor for Automatica, IEEE Control Systems Letters, and IEEE Transactions on Wireless Communications. He was also an Associate Editor for IEEE and Transactions on Signal Processing, (2007-2010, 2014-2018), IEEE Transactions on Automatic Control (2004-2007),  and Elsevier Systems and Control Letters (2003-2013).
\end{IEEEbiography}

\begin{IEEEbiography}
[{\includegraphics[width=1in,height=1.25in,clip,keepaspectratio]{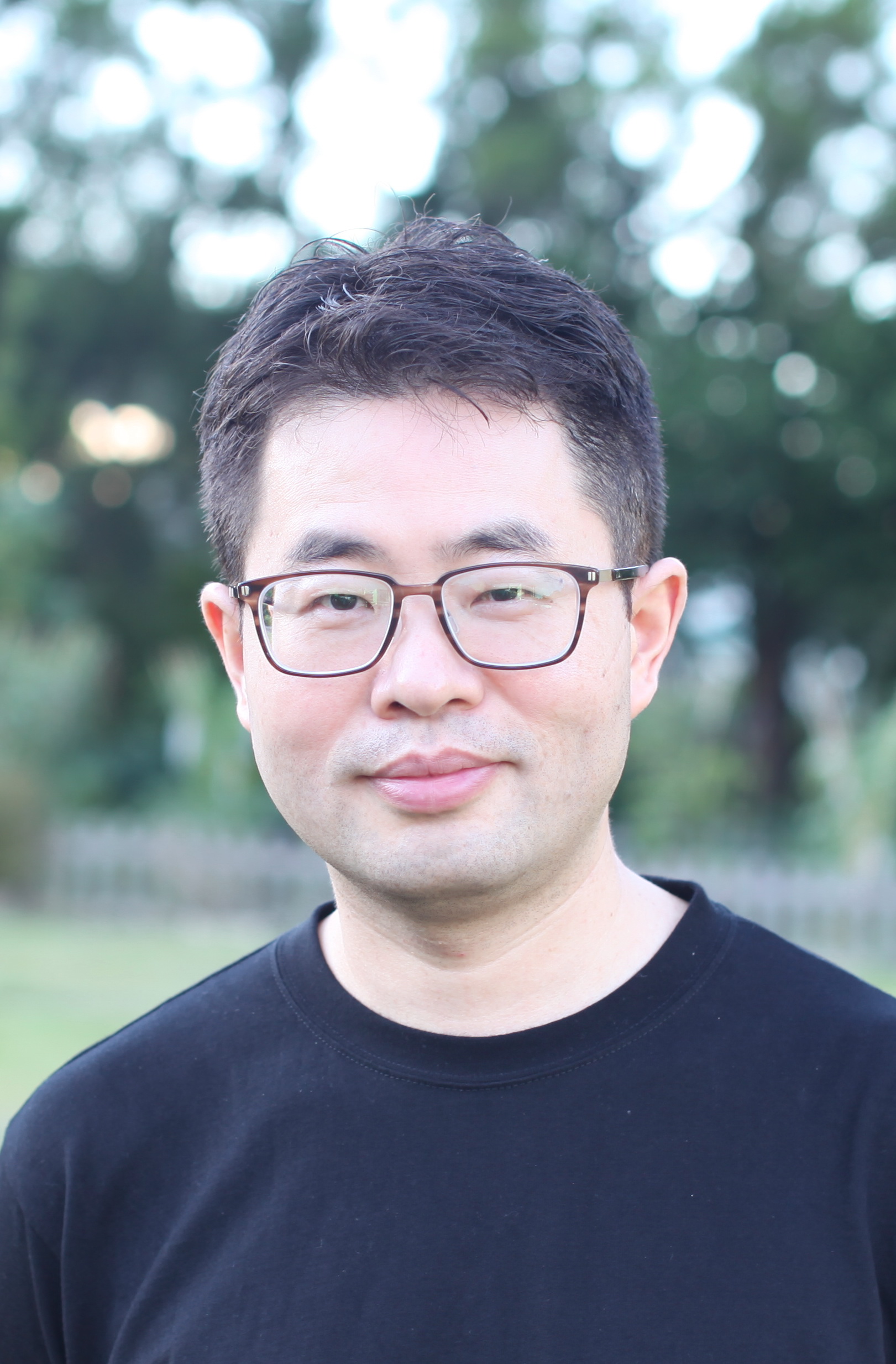}}]
{Ling Shi} received the B.E. degree in electrical and electronic engineering from Hong Kong University of Science and Technology, Kowloon, Hong Kong, in 2002 and the Ph.D. degree in Control and Dynamical Systems from California Institute of Technology, Pasadena, CA, USA, in 2008. He is currently a Professor in the Department of Electronic and Computer Engineering, and the associate director of the Robotics Institute, both at the Hong Kong University of Science and Technology. His research interests include cyber-physical systems security, networked control systems, sensor scheduling, event-based state estimation, and exoskeleton robots. He is a senior member of IEEE. He served as an editorial board member for the European Control Conference 2013-2016. He was a subject editor for International Journal of Robust and Nonlinear Control (2015-2017),  an associate editor for IEEE Transactions on Control of Network Systems (2016-2020),  an associate editor for IEEE Control Systems Letters (2017-2020), and an associate editor for a special issue on Secure Control of Cyber Physical Systems in the IEEE Transactions on Control of Network Systems (2015-2017). He also served as the General Chair of the 23rd International Symposium on Mathematical Theory of Networks and Systems (MTNS 2018). He is a member of the Young Scientists Class 2020 of the World Economic Forum (WEF).
\end{IEEEbiography}

\end{document}